\documentclass[dvipsnames,11pt]{article}
\usepackage{amsmath,amssymb,amsthm,color,bm,mathrsfs,extarrows,tikz,graphicx,mathtools,enumitem,fancybox,makecell}
\usepackage[square,sort,comma,numbers]{natbib}
\usepackage{subcaption}
\usepackage[ruled]{algorithm2e}
\usepackage[margin=1in]{geometry}
\usepackage{xcolor,bbm}
\usepackage[pagebackref]{hyperref}
\usepackage{comment}
\usepackage{soul}
\usepackage{thm-restate}
\usepackage{ifthen}
\usepackage{mathdots}
\usepackage{enumitem}

\allowdisplaybreaks

\setlist[itemize]{itemsep=0pt}
\setlist[enumerate]{itemsep=0pt}
\hypersetup{colorlinks=true,urlcolor=blue,linkcolor=blue,citecolor=[rgb]{.42,.56,.14},}
\usetikzlibrary{decorations.pathreplacing,decorations.markings,decorations.pathmorphing,decorations.shapes,arrows.meta,positioning,patterns}

\usepackage[capitalise,nameinlink]{cleveref}
\Crefname{lemma}{Lemma}{Lemmas}
\Crefname{fact}{Fact}{Facts}
\Crefname{theorem}{Theorem}{Theorems}
\Crefname{corollary}{Corollary}{Corollaries}
\Crefname{claim}{Claim}{Claims}
\Crefname{example}{Example}{Examples}
\Crefname{problem}{Problem}{Problems}
\Crefname{definition}{Definition}{Definitions}
\Crefname{notation}{Notation}{Notations}
\Crefname{assumption}{Assumption}{Assumptions}
\Crefname{subsection}{Subsection}{Subsections}
\Crefname{section}{Section}{Sections}
\Crefformat{equation}{(#2#1#3)}

\newtheorem{theorem}{Theorem}[section]
\newtheorem*{theorem*}{Theorem}

\newtheorem{proposition}[theorem]{Proposition}
\newtheorem*{proposition*}{Proposition}

\newtheorem*{property*}{Property}
\newtheorem{lemma}[theorem]{Lemma}
\newtheorem*{lemma*}{Lemma}
\newtheorem{corollary}[theorem]{Corollary}
\newtheorem*{corollary*}{Corollary}
\newtheorem*{conjecture*}{Conjecture}
\newtheorem{fact}[theorem]{Fact}
\newtheorem*{fact*}{Fact}

\newtheorem*{exercise*}{Exercise}

\newtheorem*{hypothesis*}{Hypothesis}
\newtheorem{conjecture}[theorem]{Conjecture}

\theoremstyle{definition}

\newtheorem{exercise-easy}[theorem]{Exercise}
\newtheorem{exercise-med}[theorem]{Exercise}
\newtheorem{exercise-hard}[theorem]{Exercise$^\star$}
\newtheorem{claim}[theorem]{Claim}
\newtheorem*{claim*}{Claim}

\newtheorem{remark}[theorem]{Remark}
\newtheorem*{remark*}{Remark}

\newtheorem*{observation*}{Observation}

\DeclareSymbolFont{extraup}{U}{zavm}{m}{n}
\DeclareMathSymbol{\varheart}{\mathalpha}{extraup}{86}
\DeclareMathSymbol{\vardiamond}{\mathalpha}{extraup}{87}

\DeclareMathOperator*{\E}{\mathbb E}
\DeclareMathOperator*{\Var}{\mathrm{Var}}

\DeclareMathOperator*{\argmin}{\arg\min}

\renewcommand{\Pr}{\operatorname*{\mathbf{Pr}}}

\newcommand{\Mod}[1]{\ (\mathrm{mod}\ #1)}

\newcommand{\eps}{\varepsilon}
\newcommand{\abs}[1]{\left| #1 \right|}
\newcommand{\vabs}[1]{\left\| #1 \right\|}

\newcommand{\pbra}[1]{\left( #1 \right)}
\newcommand{\sbra}[1]{\left[ #1 \right]}
\newcommand{\cbra}[1]{\left\{ #1 \right\}}
\newcommand{\floorbra}[1]{\left\lfloor #1 \right\rfloor}
\newcommand{\ceilbra}[1]{\left\lceil #1 \right\rceil}

\renewcommand{\mid}{\,\middle\vert\,}

\newcommand{\bin}{\{0,1\}}
\newcommand{\binpm}{\{\pm1\}}

\newcommand{\indicator}{\mathbbm{1}}

\newcommand{\supp}[1]{\mathsf{supp}\pbra{#1}}

\newcommand{\ac}{\mathsf{AC^0}}
\newcommand{\nc}{\mathsf{NC^0}}
\newcommand{\sym}{\mathrm{sym}}

\newcommand{\Dzeros}{\mathtt{zeros}}
\newcommand{\Dones}{\mathtt{ones}}
\newcommand{\Deven}{\mathtt{evens}}
\newcommand{\Dodd}{\mathtt{odds}}
\newcommand{\Dall}{\mathtt{all}}
\newcommand{\Dzeroones}{\mathtt{zerones}}

\newcommand{\Fbb}{\mathbb{F}}
\newcommand{\Nbb}{\mathbb{N}}
\newcommand{\Rbb}{\mathbb{R}}
\newcommand{\Zbb}{\mathbb{Z}}

\newcommand{\Dcal}{\mathcal{D}}
\newcommand{\Ecal}{\mathcal{E}}

\newcommand{\Hcal}{\mathcal{H}}
\newcommand{\Ical}{\mathcal{I}}

\newcommand{\Mcal}{\mathcal{M}}

\newcommand{\Pcal}{\mathcal{P}}
\newcommand{\Qcal}{\mathcal{Q}}

\newcommand{\Ucal}{\mathcal{U}}

\newcommand{\Wcal}{\mathcal{W}}

\newcommand{\tvdist}[1]{\vabs{#1}_\mathsf{TV}}

\renewcommand{\bar}{\overline}

\title{Locally Sampleable Uniform Symmetric Distributions}
\author{
Daniel M. Kane\thanks{University of California, San Diego. Email: \texttt{dakane@ucsd.edu}. Supported by NSF Medium Award CCF-2107547 and NSF CAREER Award CCF-1553288.}
\and
Anthony Ostuni\thanks{University of California, San Diego. Email: \texttt{aostuni@ucsd.edu}.}
\and
Kewen Wu\thanks{University of California, Berkeley. Email: \texttt{shlw\_kevin@hotmail.com}. Supported by a Sloan Research Fellowship and NSF CAREER Award CCF-2145474.}
}
\date{}

\begin{document}

\maketitle

\begin{abstract}
We characterize the power of constant-depth Boolean circuits in generating uniform symmetric distributions.
Let $f\colon\{0,1\}^m\to\{0,1\}^n$ be a Boolean function where each output bit of $f$ depends only on $O(1)$ input bits.
Assume the output distribution of $f$ on uniform input bits is close to a uniform distribution $\mathcal D$ with a symmetric support.
We show that $\mathcal D$ is essentially one of the following six possibilities: (1) point distribution on $0^n$, (2) point distribution on $1^n$, (3) uniform over $\{0^n,1^n\}$, (4) uniform over strings with even Hamming weights, (5) uniform over strings with odd Hamming weights, and (6) uniform over all strings.
This confirms a conjecture of Filmus, Leigh, Riazanov, and Sokolov (RANDOM 2023).
\end{abstract}

\section{Introduction}\label{sec:intro}

Despite being one of the simplest models of computation, $\nc$ circuits (i.e., Boolean circuits of constant depth and bounded fan-in) elude a comprehensive understanding. Even very recently, the model has been the subject of active research on the range avoidance problem \cite{ren2022range, guruswami2022range, gajulapalli2023range}, quantum advantages \cite{bravyi2018quantum, watts2019exponential, bravyi2020quantum, watts2023unconditional,kane2024locality}, proof verification \cite{goldwasser2007verifying, beyersdorff2013verifying, krebs2016small}, and more.

Pertinent to this paper is the study of the sampling power of $\nc$ circuits. While the general problem was considered at least as early as \cite{jerrum1986random}, interest in the $\nc$ setting  has seen a strong uptick lately \cite{viola2012complexity, lovett2011bounded, beck2012large, de2012extractors, viola2016quadratic, viola2020sampling, goos2020lower, chattopadhyay2022space, viola2023new, filmus2023sampling, kane2024locality, shaltiel2024explicit}.
At a high level, it considers what distributions can be (approximately) produced by simple functions on random inputs. More formally, let $f(\Ucal^m)$ denote the distribution resulting from applying an $\nc$ function $f \colon \bin^m \to \bin^n$ to a random string drawn from $\Ucal^m$, the uniform distribution over $\bin^m$. Typically, $m$ is viewed as being arbitrarily large and $n$ is the parameter of interest. Then the goal is to analyze the distance between $f(\Ucal^m)$ and some specific distribution. Aside from being inherently interesting, this question has played a crucial role in applications ranging from data structure lower bounds \cite{viola2012complexity, lovett2011bounded, beck2012large, viola2020sampling, chattopadhyay2022space, viola2023new, kane2024locality} to pseudorandom generators \cite{viola2012bit, lovett2011bounded, beck2012large} to extractors \cite{viola2012extractors, de2012extractors, viola2014extractors, chattopadhyay2016explicit, cohen2016extractors} to coding theory \cite{shaltiel2024explicit}. 

One recurring class of distributions in this line of work is uniform symmetric distributions (i.e., uniform distributions over a symmetric support). Indeed, these are precisely the distributions that arise in an elegant connection to succinct data structures (see \cite[Claim 1.8]{viola2012complexity}), for example. Moreover, this seemingly simple class is already rich enough to allow surprisingly powerful results. For instance, $\nc$ circuits can sample the uniform distribution over the preimage $\mathrm{PARITY}^{-1}(0)$ (and $\mathrm{PARITY}^{-1}(1)$), despite a celebrated result of H{\aa}stad \cite{haastad1986computational} proving that more powerful $\ac$ circuits require an exponential number of gates to \emph{compute} $\mathrm{PARITY}$. Perhaps more surprisingly, the strategy to sample a uniform random string with even Hamming weight is extremely simple: map the uniform random bits $x_1, \ldots, x_{n}$ to $x_1 \oplus x_2, x_2 \oplus x_3, \ldots, x_{n} \oplus x_1$ \cite{babai1987random, boppana1987one}. 

A number of notable prior results already rule out specific distributions from being accurately sampled by $\nc$ circuits. For example, let $\Dcal_{\cbra{k}}$ denote the uniform distribution over all $n$-bit strings of Hamming weight $k$. The influential early paper of \cite{viola2012complexity} showed that such shallow circuits could not accurately sample $\Dcal_{\cbra{k}}$ for $k = \Theta(n)$ under certain assumptions about the input length or accuracy tolerance; recent works \cite{filmus2023sampling, viola2023new, kane2024locality} have eliminated the need for these assumptions. Additionally, a number of results are known for uniform symmetric distributions over multiple Hamming weights, such as the case of exclusively tail weights \cite{filmus2023sampling}, all weights divisible by $q$ for fixed $3 \le q \ll \sqrt{n}$ \cite{kane2024locality}, and all weights above $n/2$ \cite{goldwasser2007verifying,viola2012complexity,filmus2023sampling} (see also \cite{watts2023unconditional}).

Despite much effort, the previous body of work proceeds in a somewhat ad-hoc fashion, with techniques tailored to rule out specific cases. However, an exciting recent work by Filmus, Leigh, Riazanov, and Sokolov \cite{filmus2023sampling} gave the following bold conjecture about the capabilities of $\nc$ circuits for sampling distributions, unifying prior results. 
\begin{conjecture}
[{\cite[Conjecture 1.1]{filmus2023sampling}}]
\label{conj:nc0_classification}
    For every $d \in \Nbb, \eps \in (0,1)$, and large enough $n$, if $f\colon\bin^m\to\bin^n$ is computable by an $\nc$ circuit of depth at most $d$ and $f(\Ucal^m)$ is $\eps$-close (in total variation distance) to a uniform symmetric distribution, then $f(\Ucal^m)$ is $O(\eps)$-close to one of the following six distributions:
    \begin{itemize}
        \item Point distribution on $0^n$.
        \item Point distribution on $1^n$.
        \item Uniform distribution over $\cbra{0^n,1^n}$.
        \item Uniform distribution over strings with even Hamming weights.
        \item Uniform distribution over strings with odd Hamming weights.
        \item Uniform distribution over all strings.
    \end{itemize}
\end{conjecture}
All six distributions can be sampled (exactly) by functions whose output bits each depend on at most two input bits. Hence one may informally view the conjecture as asserting that more input dependencies do not substantially increase the ability of $\nc$ circuits to generate uniform symmetric distributions.

In this work, we confirm the conjecture of \cite{filmus2023sampling}.

\begin{theorem}[Consequence of \Cref{thm:nc0}]\label{thm:nc0_informal}
    \Cref{conj:nc0_classification} is true.
\end{theorem}

We emphasize that the implicit constant in the distance $O(\eps)$ in our result has \emph{no} dependence on the depth $d$. Additionally, note that this result is optimal up to that implicit constant. We include a more thorough discussion of our result's tightness in \Cref{sec:nc0}, where we present a quantitative version of \Cref{thm:nc0_informal} parametrized by the locality (i.e., number of input bits each output bit depends on) of $f$. The following corollary is immediate.
\begin{corollary}
    For sufficiently large $n$, the only uniform symmetric distributions over $\bin^n$ exactly sampleable by $\nc$ functions are the six distributions in \Cref{conj:nc0_classification}.
\end{corollary}

As a contrasting example to the limitation given by \Cref{thm:nc0_informal}, consider the next simplest class of circuits commonly studied: $\ac$. Up to some exponentially small error, they are able to sample the uniform distribution over permutations of $[n]$ \cite{matias1991converting, hagerup1991fast}. Thus by sampling $1^w 0^{n-w}$ for the appropriate distribution over weights $w$ accepted (or rejected) by a symmetric function $f$, one can apply a randomly sampled permutation to output the uniform distribution over $f^{-1}(1)$ (or $f^{-1}(0)$) \cite[Lemma 4.3]{viola2012complexity}.

\paragraph*{Paper Organization.}
We provide a proof overview of \Cref{thm:nc0_informal} in \Cref{sec:overview}. Preliminary definitions and results are given in \Cref{sec:prelim}. The full proof of our main result is in \Cref{sec:nc0}, with some technical proofs deferred to \Cref{sec:LLT} and \Cref{app:comp_llt}.
\section{Proof Overview}\label{sec:overview}

Our starting point is similar to many past works \cite{viola2012complexity, viola2020sampling, filmus2023sampling, viola2023new, kane2024locality}: we reduce an arbitrary function (computable by an $\nc$ circuit) to a collection of structured functions which are more amenable to analysis. Our results then follow by lifting insights from these structured functions to our original function.

It will be convenient to work with the abstraction of \emph{locality}. We say a function $f\colon \bin^m \to \bin^n$ is \emph{$d$-local} if every output bit depends on at most $d$ input bits. Observe that the class of $d$-local functions captures functions computable by Boolean circuits of depth $O(\log d)$ and bounded fan-in. In particular, constant locality functions are equivalent to those computable by $\nc$ circuits.
Henceforth, let $f\colon \bin^m \to \bin^n$ be a $d$-local function. For simplicity, we hide minor factors in the following discussion.

\subsection{A Structured Decomposition}\label{sec:KOW_summary}

We will use the ``graph elimination'' reduction strategy of \cite{kane2024locality}, rephrased slightly more naturally in the language of hypergraphs. Let $G$ be the hypergraph on the output bits $[n]$ with an edge for each input bit connecting all of the output bits that depend on it. By assumption, each vertex is contained in at most $d$ edges (i.e.\ $G$ has maximum degree $d$). Using standard hypergraph terminology, we define the \emph{neighborhood} of a vertex $v$ to be the set of all vertices sharing an edge with $v$. Furthermore, we call two neighborhoods $N_1, N_2$ \emph{connected} if there exist vertices $v_1 \in N_1, v_2 \in N_2$ contained in the same edge.

By \cite[Proposition 5.20]{kane2024locality}, there exists a set of $o(n)$ edges whose deletion results in a graph with $\Omega_d(n)$ non-connected neighborhoods of size $O_d(1)$. In other words, there always exists a choice of a few input bits whose conditioning upon decomposes $f$ into a mixture of subfunctions with substantially independent output bits.

This independence is crucial in ruling out the sampleability of various distributions by these structured subfunctions. For example, \cite{kane2024locality} used the following win-win argument to prove strong bounds on the distance between any distribution sampleable by a local function and the uniform distribution over $n$-bit strings of Hamming weight $k=\Theta(n)$, denoted $\Dcal_{\cbra{k}}$. If the marginal distributions of most independent neighborhoods noticeably differ from the corresponding marginals of $\Dcal_{\cbra{k}}$, then the errors can be combined together via a straightforward concentration bound argument \cite[Lemma 4.2]{kane2024locality}. 

Otherwise, the marginal distributions of most independent neighborhoods $N(v_1), \dots, N(v_r)$ closely match the marginals of $\Dcal_{\cbra{k}}$. Hence by conditioning on all the input bits that do not affect the output bits $v_1,\dots, v_r$, the weight of the output becomes a sum of well-behaved independent integer random variables. From this property, one can show (\cite[Claims 5.16 \& 5.23]{kane2024locality}) that with high probability many of these random variables are not constant (or even constant modulo $q$ for $q \ge 3$), in which case anticoncentration inequalities (e.g., \cite[Theorem 3]{ushakov1986upper} or \cite[Lemma 3.7]{kane2024locality}) imply no specific output weight can be obtained with good probability. Hence the subfunctions cannot accurately sample $\Dcal_{\cbra{k}}$, so (by a union bound argument) neither can their mixture.

Note that the distribution $\Dcal_{\cbra{k}}$ is a special kind of uniform symmetric distribution (i.e., uniform distribution over a symmetric support).
In this work, we need to handle more general ones; however, many of the same ideas will drive our analysis.

\subsection{Classification of Locally Sampleable Uniform Symmetric Distributions}\label{sec:overview_nc0}

Now we show how to handle a general uniform symmetric distribution and obtain our classification result.
For convenience, we use a non-empty set $\Psi\subseteq\cbra{0,1,\ldots,n}$ to denote the acceptable Hamming weights and use $\Dcal_\Psi$ to denote the uniform distribution over strings of Hamming weights in $\Psi$.
Then our goal is to show that local functions cannot approximate $\Dcal_\Psi$ unless $\Psi$ is $\cbra{0}$ (the point distribution on $0^n$), $\cbra{n}$ (the point distribution on $1^n$), $\cbra{0,n}$ (uniform over $\cbra{0^n,1^n}$), $\cbra{0,2,4,\ldots}$ (uniform over strings with even parity), $\cbra{1,3,5,\ldots}$ (uniform over strings with odd parity), or $\cbra{0,1,\ldots,n}$ (uniform over all strings). We will often refer to the corresponding $\Dcal_\Psi$ as the six special distributions.

Let $s\in\Psi$ be an element closest to the middle weight $n/2$. Note the majority of the mass of $\Dcal_\Psi$ is supported on strings roughly as close to $n/2$ as $s$ is.
Informally, we view $\Dcal_\Psi$ as either $\Dcal_{\cbra{s}}$ (uniform over the Hamming slice of weight $s$) or $\frac12\Dcal_{\cbra{s}}+\frac12\Dcal_{\cbra{n-s}}$ (uniform over the Hamming slices of weight $s$ and $n-s$).
Then the above six locally sampleable distributions can be classified by $s$: either it is the endpoint (i.e., $s$ equals $0$ or $n$) or it is the middle point (i.e., $s$ roughly equals $n/2$).
Our proof follows this intuition. If $|s - n/2| > n^{2/3}$, we will show that it must be the case of $s\in\cbra{0,n}$. 
Otherwise $|s - n/2| \le n^{2/3}$, and we will show that it must be the case that $\Psi$ is effectively all-even, all-odd, or everything. 

\paragraph*{The $|s-n/2| > n^{2/3}$ Case.}
This far in the tails the binomial distribution decays rapidly, so it is not difficult to show (\Cref{clm:small_support}) that half the mass of $\Dcal_\Psi$ is supported on $O(n^{1/3})$ different weights.
However, \cite[Theorem 1.2]{filmus2023sampling} and \cite[Theorem 5.10]{kane2024locality} show that any distribution $\Dcal_{\cbra{k}}$ other than $k=0$ or $n$ has total variation distance $1-O_d(n^{-1/2})$ from any $d$-locally sampleable uniform symmetric distribution.
Therefore, by a union bound argument (\Cref{lem:tvdist_after_conditioning}), $\tvdist{f(\Ucal^m) - \Dcal_\Psi} > 1/3$ unless $\Psi$ is one of $\{0\},\{n\},\{0,n\}$.
We now turn to the more challenging case.

\paragraph*{The $|s-n/2| \le n^{2/3}$ Case.}
In this case we note that $\Dcal_\Psi$ is reasonably close to the uniform distribution. (In particular, its restriction to any constant number of output bits will be quite close.) Thus, the framework explained in \Cref{sec:KOW_summary} is now applicable. That is, we remove a modest number of edges from the corresponding hypergraph to obtain one with $\Omega_d(n)$ many non-adjacent neighborhoods of size $O_d(1)$. This means that after conditioning on the removed inputs, we will have many sets of output bits that are independent of each other.

From here we again split into cases. 
If many of these neighborhoods are far from uniform, our distribution must be far from the uniform-ish $\Dcal_\Psi$. 
This part of the mixture will contribute nearly equally to the distance from $\Dcal_\Psi$ and from the closest special distribution. 
Otherwise, we have many roughly unbiased neighborhoods $N(v_1),\dots,N(v_r)$. 
As sketched in \Cref{sec:KOW_summary}, if we further condition on all the inputs that do not affect $v_r, \dots, v_r$, then the total weight of the output becomes a sum of independent random variables. 
Using the fact that the overall distribution on these neighborhoods is uniform, we can show that with high probability many of these random variables are not constant, and in fact for each $q>2$, many are not constant modulo $q$. 
This allows us to show (\Cref{continuity prop}) that our distribution over weights must be continuous. 
In particular, the probability that we have weight $x$ and the probability that we have weight $x+\Delta$ for even $\Delta \in \Zbb$ must differ by at most $O_d(|\Delta|/n)$ (\Cref{thm:comp_llt_C=1_special}). 
This implies that if $f(\Ucal^m)$ is close to $\Dcal_\Psi$, then $\Psi$ cannot have many instances where it contains $x$ but not $x+\Delta$ or vice versa. 
By pairing up the present $x$'s with missing ones, we attempt to show that the distance between $f(\Ucal^m)$ with $\Dcal_\Psi$ is comparable to the distance between $\Dcal_\Psi$ and the nearest special distribution.

Unfortunately, there are two issues with this argument. The first is that our bounds on the difference between probabilities of $x$ and $x+\Delta$ have error terms that are inverse polynomial in $n$.\footnote{The following issue persists even if we sharpen the bounds to have inverse exponential error (as was done in a previous version \url{https://arxiv.org/pdf/2411.08183v1}).} 
This makes them useless if $\Psi$ contains all of the elements of $[n]$ within $\Omega_d(n)$ of $n/2$. 
In this case, however, the coordinates of $f(\Ucal^m)$ must match moments with the uniform distribution (\Cref{moment matching lemma}). 
Thus, a careful argument (\Cref{prop:mass_from_S}) can show that $f(\Ucal^m)$ cannot be much closer to $\Dcal_\Psi$ than the relevant special distribution, as a decent fraction of any mass removed from the first ``missing'' Hamming slice $\Dcal_{\cbra{s}}$ of $\Dcal_\Psi$ must be placed at even more extreme weights rather than closer to the middle.

The other issue is simply the $d$-dependence of our continuity bounds. A naive application would only be enough to show that $f(\Ucal^m)$ is $O_d(\eps)$-close to the nearest special distribution. 
If, for example, $\Psi = \{0,1,\ldots,n/2\}$, we would only be able to find $\Omega_d(\sqrt{n})$ many pairs of elements to compare and could only show that $\tvdist{f(\Ucal^m) - \Dcal_\Psi}=\Omega_d(1).$ 
To fix this, we use the fact that our output distribution must match moments with the uniform distribution to show (\Cref{Kol prop}) that most of $f(\Ucal^m)$'s mass has a weight distribution close in Kolmogorov distance\footnote{Recall the \emph{Kolmogorov distance} between two distributions $\Pcal$, $\Qcal$ is given by $\sup_{t \in \Rbb}|\Pr_{x\in \Pcal}[x > t] - \Pr_{y\in \Qcal}[y > t]|$.} to the binomial distribution, even accounting for parity. 

More specifically, we use an inductive argument (\Cref{tree lemma}) to decompose the input space $\bin^m$ into subcubes. 
Aside from some negligible mass, each subcube falls into two types. A Type-I subcube $C$ satisfies that the output bits of $f$ on the uniform distribution over $C$ are nearly $k$-wise independent (for an appropriately chosen parameter $k$) with no input bit affecting too many output bits. Such cubes are useful, as we can leverage known results about $k$-wise independence fooling threshold functions (\Cref{thm:kwise_ind}) and the structure of low-degree $\Fbb_2$-polynomials (\Cref{parity randomization theorem}) to prove our desired Kolmogorov property in this special case (\Cref{k indep lem}). 

Alternatively, Type-II subcubes satisfy that $f$ applied to the uniform distribution over the union of all Type-II subcubes is far from $\Dcal_\Psi$. 
These correspond to the aforementioned neighborhoods whose distributions are far from uniform, and whose contribution to the distance between $f(\Ucal^m)$ and $\Dcal_\Psi$ roughly matches that of the contribution to the distance between $f(\Ucal^m)$ and the nearest special distribution.
This is the small fraction of $f(\Ucal^m)$'s mass that does not satisfy the desired Kolmogorov distance property.

Combining with our continuity result, we obtain a local limit theorem (\Cref{LLT Theorem}) stating that $f(\Ucal^m)$ must be close to some mixture $\Mcal$ of the uniform distribution over strings with even Hamming weight and the uniform distribution over strings with odd Hamming weight. 
Furthermore, observe that for any two strings, $\Dcal_\Psi$ either assigns them the same mass, or assigns one of them zero mass. Since $f(\Ucal^m)$ is close to both $\Dcal_\Psi$ and $\Mcal$, the triangle inequality guarantees $\Dcal_\Psi$ is close to $\Mcal$. Then by the previous observation, $\Mcal$ must be close to either the uniform distribution over strings with even Hamming weight, the uniform distribution over strings with odd Hamming weight, or the uniform distribution over all strings (\Cref{close special cor}). In other words, we can extract a similar conclusion to our limit limit theorem for the nearest special distribution $\Dcal$: $\tvdist{f(\Ucal^m) - \Dcal} \leq O(\tvdist{f(\Ucal^m) -\Dcal_\Psi}) + \delta$ for any $\delta > 0$, so long as $n$ is sufficiently large in terms of $\delta$ and $d$.

At this point, it is tempting to apply the above conclusion with $\delta = O(\tvdist{f(\Ucal^m) -\Dcal_\Psi})$ to conclude the proof. 
However, if $\delta$ is particularly small, we may not be able to guarantee $n$ is large enough to apply the result. 
Hence, we can only assume $\tvdist{f(\Ucal^m) - \Dcal} \leq O(\tvdist{f(\Ucal^m) -\Dcal_\Psi}) + \delta$ for some sufficiently small constant $\delta > 0$.

It remains to handle the special case of $\tvdist{f(\Ucal^m) - \Dcal}$ less than some sufficiently small constant. 
Here, we can take advantage of strong structural properties, such as matching many moments with the uniform distribution (\Cref{moment matching lemma}) and having the support of $f(\Ucal^m)$ entirely contained in the support of $\Dcal$ (\Cref{constant parity lem}). 
These additional properties allow us to deduce that unless our desired result already holds, the weight distribution $|f(\Ucal^m)|$ has noticeably too little mass on weights around $n/2$. We can now run the previously discussed pairing argument, where we find nearby $x\in \Psi$ and $x+\Delta \not\in\Psi$ for even $\Delta \in \Zbb$, except now restricted to a small interval around $n/2$. Crucially, the mass $\Dcal_\Psi$ assigns to strings with weight in that interval exceeds the error from the continuity bound, removing the dependence on $d$.

\paragraph*{Previous Version.}
A previously posted version of this paper provided a weaker form of \Cref{thm:nc0_informal} where $f(\Ucal^m)$ is only $O_d(\eps)$-close to one of the six special distributions. The approach taken was a more technically involved form of the original pairing argument described above. In particular, the key improvement in this version that allows us to remove the dependence on $d$ from the distance bound is the local limit theorem (\Cref{LLT Theorem}).

\section{Preliminaries}\label{sec:prelim}

For a positive integer $n$, we use $[n]$ to denote the set $\cbra{1,2,\ldots,n}$.
We use $\Rbb$ to denote the set of real numbers, use $\Nbb=\cbra{0,1,2,\ldots}$ to denote the set of natural numbers, and use $\Zbb$ to denote the set of integers.
For a binary string $x$, we use $|x|$ to denote its Hamming weight.
We use $\log(x)$ and $\ln(x)$ to denote the logarithm with base $2$ and $e\approx2.71828\ldots$ respectively.

\paragraph*{Asymptotics.}
We use the standard $O(\cdot), \Omega(\cdot), \Theta(\cdot)$ notation, and emphasize that in this paper they only hide universal positive constants that do not depend on any parameter. Occasionally we will use subscripts to suppress a dependence on particular variable (e.g.,\ $O_d(1)$).

\paragraph*{Probability.}
Let $\Pcal$ be a (discrete) distribution. We use $x\sim\Pcal$ to denote a random sample $x$ drawn from the distribution $\Pcal$.
If $\Pcal$ is a distribution over a product space, then we say $\Pcal$ is a product distribution if its coordinates are independent.
In addition, let $S$ be a non-empty set. If $S$ indexes $\Pcal$, we use $\Pcal[S]$ to denote the marginal distribution of $\Pcal$ on coordinates in $S$. We reserve $\Ucal$ to denote the uniform distribution over $\bin$ and $\Ucal(S)$ to denote the uniform distribution over $S$.

For a deterministic function $f$, we use $f(\Pcal)$ to denote the output distribution of $f(x)$ given a random $x\sim\Pcal$. 
For every event $\Ecal$, we define $\Pcal(\Ecal)$ to be the probability that $\Ecal$ happens under distribution $\Pcal$.
In addition, we use $\Pcal(x)$ to denote the probability mass of $x$ under $\Pcal$, and use $\supp{\Pcal}=\cbra{x:\Pcal(x)>0}$ to denote the support of $\Pcal$.
If $\Pcal$ is a distribution over $\bin^n$, we use $|\Pcal|$ to denote the distribution over weights. That is, $|\Pcal|(w) = \sum_{x : |x|=w}\Pcal(x)$.
We additionally define the symmetrized distribution $\Pcal_\sym$ to be the distribution resulting from randomly permuting strings $x\sim \Pcal$.

Let $\Qcal$ be a distribution. We use $\tvdist{\Pcal-\Qcal}=\frac12\sum_x\abs{\Pcal(x)-\Qcal(x)}$ to denote their total variation distance.\footnote{To evaluate total variation distance, we need two distributions to have the same sample space. This will be clear throughout the paper and thus we omit it for simplicity.}
We say $\Pcal$ is $\eps$-close to $\Qcal$ if $\tvdist{\Pcal(x)-\Qcal(x)}\le\eps$, and $\eps$-far otherwise.

\begin{fact}\label{fct:tvdist}
Total variation distance has the following equivalent characterizations:
$$
\tvdist{\Pcal-\Qcal}=\max_{\text{event }\Ecal}\Pcal(\Ecal)-\Qcal(\Ecal)=\min_{\substack{\text{random variable }(X,Y)\\\text{$X$ has marginal $\Pcal$ and $Y$ has marginal $\Qcal$}}}\Pr\sbra{X\neq Y}.
$$
\end{fact}

Let $\Pcal_1,\ldots,\Pcal_t$ be distributions.
Then $\Pcal_1\times\cdots\times\Pcal_t$ is a distribution denoting the product of $\Pcal_1,\ldots,\Pcal_t$.
We also use $\Pcal^t$ to denote $\Pcal_1\times\cdots\times\Pcal_t$ if each $\Pcal_i$ is the same as $\Pcal$.
For a finite set $S \subseteq [t]$, we use $\Pcal^S$ to denote the distribution $\Pcal^t$ restricted to the coordinates of $S$.
We say distribution $\Pcal$ is a convex combination (or mixture) of $\Pcal_1,\ldots,\Pcal_t$ if there exist $\alpha_1,\ldots,\alpha_t\in[0,1]$ such that $\sum_{i\in[t]}\alpha_i=1$ and $\Pcal(x)=\sum_{i\in[t]}\alpha_i\cdot\Pcal_i(x)$ for all $x$ in the sample space. When it is clear from context, we will occasionally write mixtures more simply as $\Pcal=\sum_{i\in[t]}\alpha_i\cdot\Pcal_i$.

We will require two inequalities about total variation distance.
The first allows us to argue that two product distributions must be far apart if their marginals do not match.

\begin{lemma}[{\cite[Lemma 4.2]{kane2024locality}}]\label{lem:tvdist_after_product}
Let $\Pcal$, $\Qcal$, and $\Wcal$ be distributions over an $n$-dimensional product space, and let $S\subseteq[n]$ be a non-empty set of size $s$.
Assume
\begin{itemize}
\item $\Pcal[S]$ and $\Wcal[S]$ are two product distributions,
\item $\tvdist{\Pcal[\cbra{i}]-\Wcal[\cbra{i}]}\ge\eps$ holds for all $i\in S$, and
\item $\Wcal(x)\ge\eta\cdot\Qcal(x)$ holds for some $\eta>0$ and all $x$.
\end{itemize}
Then
$$
\tvdist{\Pcal-\Qcal}\ge1-2\cdot e^{-\eps^2s/2}/\eta.
$$
\end{lemma}

The second states that the distance between a distribution $\Dcal$ and a mixture of distributions must be large if the distance between $\Dcal$ and each individual distribution in the mixture is also large.

\begin{lemma}[{\cite[Section 4.1]{viola2020sampling}, \cite[Lemma 4.3]{kane2024locality}}]\label{lem:tvdist_after_conditioning}
Let $\Pcal_1,\ldots,\Pcal_t$ and $\Qcal$ be distributions.
Assume there exists a value $\eps$ such that $\tvdist{\Pcal_i-\Qcal}\ge1-\eps$ for all $i\in [t]$.
Then for any distribution $\Pcal$ as a convex combination of $\Pcal_1,\ldots,\Pcal_t$, we have
$$
\tvdist{\Pcal-\Qcal}\ge1-(t+1)\cdot\eps.
$$
\end{lemma}
Occasionally, we will use a different distance measure between distributions $\Pcal$ and $\Qcal$; namely, the Kolmogorov distance
\[
    \max_t \big| \Pr_{x \sim \Pcal}\sbra{x > t} - \Pr_{y\sim \Qcal}\sbra{y > t}\big|.
\]

\paragraph*{Locality and Hypergraphs.}
Let $f\colon\bin^m\to\bin^n$. We say $f$ is a $d$-local function if each output bit $i\in [n]$ depends on at most $d$ output bits. If unspecified, we will always assume $n,m,d$ are positive integers.
 
We sometimes take an alternative view, using hypergraphs to model the dependency relations in $f$.
Let $G=(V,E)$ be an (undirected) hypergraph.
For each $i\in V$, we use $I_G(i)\subseteq V$ to denote the set of vertices that share an edge with $i$.
We say $G$ has maximum degree $d$ if $|I_G(i)|\le d$ holds for all $i\in V$.
Define $N_G(i)=\cbra{i'\in V\colon I_G(i)\cap I_G(i')\neq\emptyset}$ to be the neighborhood of $i$. We visualize the function $f\colon\bin^m\to\bin^n$ as a hypergraph on the output bits $[n]$ with an edge for each input bit containing all of the output
bits that depend on it.

\paragraph*{Concentration.}
We will need the following standard concentration bounds.
\begin{fact}[Hoeffding's Inequality]\label{fct:hoeffding}
Assume $X_1,\ldots,X_n$ are independent random variables such that $a\le X_i\le b$ holds for all $i\in[n]$.
Then for all $\delta\ge0$, we have
$$
\max\cbra{\Pr\sbra{\frac1n\sum_{i\in[n]}\pbra{X_i-\E[X_i]}\ge\delta},\Pr\sbra{\frac1n\sum_{i\in[n]}\pbra{X_i-\E[X_i]}\le-\delta}}
\le\exp\cbra{-\frac{2n\delta^2}{(b-a)^2}}.
$$
\end{fact}

\begin{fact}[Chernoff's Inequality]\label{fct:chernoff}
Assume $X_1,\ldots,X_n$ are independent random variables such that $X_i\in[0,1]$ holds for all $i\in[n]$.
Let $\mu=\sum_{i\in[n]}\E[X_i]$.
Then for all $\delta\in[0,1]$, we have
$$
\Pr\sbra{\sum_{i\in[n]}X_i\le(1-\delta)\mu}
\le\exp\cbra{-\frac{\delta^2\mu}2}.
$$
\end{fact}

\paragraph*{Norms and Hypercontractivity.} 
For a function $g\colon \bin^n \to \bin$ and integer $q \ge 1$, define the norm $\|g\|_q = \pbra{\E_{x \sim \bin^n}[|g(x)|^q]}^{1/q}$. We will use the following \emph{hypercontractive inequality} in several locations to control higher norms.
\begin{lemma}[\cite{bonami1970etude}]\label{lem:hypercontractivity}
    For a degree-$d$ polynomial $p\colon \bin^n \to \mathbb{R}$ and an integer $q \ge 2$, we have
    \[
        \|p\|_q \le \sqrt{q-1}^d \cdot \|p\|_2.
    \]
\end{lemma}

\paragraph*{Binomials and Entropy.}
Let $\Hcal(x)=x\cdot\log\pbra{\frac1x}+(1-x)\cdot\log\pbra{\frac1{1-x}}$ be the binary entropy function.
We will frequently use the following estimates regarding binomial coefficients and the entropy function.

\begin{fact}[{See e.g., \cite[Lemma 17.5.1]{cover2006elements}}]\label{fct:individual_binom}
For $1\le k\le n-1$, we have
$$
\frac{2^{n\cdot\Hcal(k/n)}}{\sqrt{8k(1-k/n)}}\le\binom nk\le\frac{2^{n\cdot\Hcal(k/n)}}{\sqrt{\pi k(1-k/n)}}.
$$
\end{fact}

\begin{fact}[{See e.g., \cite{wiki:Binary_entropy_function}}]\label{fct:entropy}
For any $x\in[-1,1]$, we have
$$
1-x^2\le\Hcal\pbra{\frac{1+x}2}=1-\frac1{2\ln(2)}\sum_{n=1}^{+\infty}\frac{x^{2n}}{n\cdot(2n-1)}\le1-\frac{x^2}{2\ln(2)}.
$$
\end{fact}

\begin{fact}[{See e.g., \cite{wiki:Binomial_coefficient,binomial_sum_mo}}]\label{fct:binom_tail_asym}
For $1\le k\le n/2$, we have
$$
\sum_{i=0}^k\binom ni\le\min\cbra{2^{n\cdot\Hcal(k/n)},\binom nk\cdot\frac{n-k+1}{n-2k+1}}.
$$
\end{fact}

\section{Classification of Locally Sampleable Distributions}\label{sec:nc0}

In this section, we will prove a general classification result for uniform distributions with symmetric support that can be sampled by local functions.
Let $\Psi\subseteq\cbra{0,1,\ldots,n}$ be a non-empty set. We define $\Dcal_\Psi$ to be the uniform distribution over $x\in\bin^n$ conditioned on $|x|\in\Psi$.

We will show that if the output distribution of a local function is close to $\Dcal_\Psi$, then it is in fact close to one of the following six specific symmetric distributions: $\Dzeros$, $\Dones$, $\Dzeroones$, $\Deven$, $\Dodd$, and $\Dall$, where
\begin{itemize}
\item $\Dzeros=\Dcal_{\cbra{0}}$, $\Dones=\Dcal_{\cbra{n}}$, and $\Dzeroones=\Dcal_{\cbra{0,n}}$,
\item $\Deven=\Dcal_{\cbra{\text{even numbers in }\cbra{0,1,\ldots,n}}}$ and $\Dodd=\Dcal_{\cbra{\text{odd numbers in }\cbra{0,1,\ldots,n}}}$,
\item $\Dall=\Dcal_{\cbra{0,1,\ldots,n}}$.
\end{itemize}

\begin{theorem}\label{thm:nc0}
Let $d \in \Nbb$ and $f\colon\bin^m\to\bin^n$ be a $d$-local function with $n$ sufficiently large (in terms of $d$).
Assume $\tvdist{f(\Ucal^m)-\Dcal_\Psi}\le\eps$ for some $\Psi\subseteq\cbra{0,1,\ldots,n}$.
Then
$$
    \tvdist{f(\Ucal^m)-\Dcal}\le O(\eps)
$$
for some $\Dcal\in\cbra{\Dzeros,\Dones,\Dzeroones,\Deven,\Dodd,\Dall}$.
\end{theorem}

\begin{remark}\label{rmk:nc0}
We show the qualitative tightness of \Cref{thm:nc0} from different angles.
\begin{itemize}
\item The six special distributions admit local sampling schemes: $\Dzeros$ and $\Dones$ can be sampled by a $0$-local function; $\Dall$ and $\Dzeroones$ can be sampled by a $1$-local function; $\Deven$ and $\Dodd$ can be sampled by a $2$-local function.
\item The lower bound on $n$ is necessarily depending on $d$. 
If $n\le d$, then one can sample the uniform distribution over any support $S\subseteq\bin^n$ of size $|S|$ dividing $2^d$. This can be achieved by fixing a regular mapping $\pi\colon\bin^d\to S$ and using the $d$ input bits to compute it.
Also if $n$ is a power of two and $d=\log(n)$, then one can directly sample a uniform string of Hamming weight one, which is uniform symmetric. 
\item The unspecified distance assumption $\eps$ cannot be replaced by a constant, i.e., local functions are indeed able to \emph{arbitrarily} closely approximate uniform symmetric distributions.

Starting from $\Deven$, we randomly flip the first $c\in[n]$ output bits with probability $1/4$. 
This distribution is $4$-local since both $\Deven$ and the $1/4$-biased flipping are $2$-local.
It is easy to see that this distribution is at distance $2^{-\Theta(c)}$ to $\Dall$ and $\Deven$, and is much farther from other uniform symmetric distributions.
This shows that $\eps$ can be arbitrarily small even when $d$ is a fixed constant.

\item The distance blowup from $\tvdist{f(\Ucal^m)-\Dcal_\Psi}$ to $\tvdist{f(\Ucal^m)-\Dcal}$ is qualitatively necessary, i.e., the local distribution can be closer to a uniform symmetric distribution than to one of the six special ones. In particular, we identify the following example which rules out a bound of the form $(1+o(1))\eps$.

Consider the distribution $\Dcal$ that with probability $3/4$ is $\Deven$ and with probability $1/4$ is $\Dodd$. Observe $\Dcal$ can be sampled by a 3-local function via a similar strategy to that for $\Deven$. The uniform distribution over $n$-bit strings of Hamming weight $0, 1, 2,$ or $4 \bmod 6$ is approximately $(1/6)$-close to $\Dcal$; however, all six special distributions are at least (1/4)-far from $\Dcal$. Thus, the implicit constant in \Cref{thm:nc0} must be at least $3/2$.
\end{itemize}
\end{remark}

To prove \Cref{thm:nc0}, we will classify $\Psi$ into two cases and handle them separately.
To this end, define $\iota(\Psi)\in \Psi$ to be the Hamming weight in $\Psi$ that is closest to the middle: 
$$
\iota(\Psi)=\argmin_{s\in \Psi}\abs{s-n/2}
$$ 
where we break ties arbitrarily.
Intuitively, since $\iota(\Psi)$ is the dominating Hamming weight under the binomial distribution, $\Dcal_\Psi$ is close to either $\Dcal_{\cbra{\iota(\Psi)}}$ or $\frac12\pbra{\Dcal_{\cbra{\iota(\Psi)}}+\Dcal_{\cbra{n-\iota(\Psi)}}}$.
Based on this intuition, we divide $\Psi$ into the following cases:

\begin{itemize}
\item \textsc{Tail Regime}:
$\iota(\Psi) < n/2 - n^{2/3}$ or $\iota(\Psi) > n/2 + n^{2/3}$.
\item \textsc{Central Regime}:
$n/2 - n^{2/3}\le \iota(\Psi)\le n/2 + n^{2/3}$.
\end{itemize}

In the tail regime, we wish to show that $\Dcal_\Psi$ is essentially $\Dzeros$, $\Dones$, or $\Dzeroones$. 
The following result will be proved in \Cref{sec:middle_regime}.
\begin{theorem}\label{thm:middle_regime}
    Let $f \colon \bin^m \to \bin^n$ be a $d$-local function with $n$ sufficiently large (in terms of $d$). If  $\tvdist{f(\Ucal^m) - \Dcal_\Psi} \le 1/3$ for some $\Psi$ in the tail regime, then $\Dcal_\Psi \in \cbra{\Dzeros, \Dones, \Dzeroones}$.
\end{theorem}

In the central regime, we aim to show that $\Dcal_\Psi$ is essentially $\Deven$, $\Dodd$, or $\Dall$ by \Cref{thm:central_regime}, which will be proved in \Cref{sec:central_regime}.

\begin{theorem}\label{thm:central_regime}
    Let $f \colon \bin^m \to \bin^n$ be a $d$-local function with $n$ sufficiently large (in terms of $d$). If $\tvdist{f(\Ucal^m) - \Dcal_\Psi} \le \eps$ for some $\Psi$ in the central regime, then $\tvdist{f(\Ucal^m) - \Dcal} \le O(\eps)$ for some $\Dcal \in \cbra{\Deven, \Dodd, \Dall}$.
\end{theorem}

We can now easily establish \Cref{thm:nc0}.
\begin{proof}[Proof of \Cref{thm:nc0}]
    We assume $d \ge 1$, as otherwise $f(\Ucal^m)$ is $\Omega(1)$-far from every uniform symmetric distribution other than $\Dzeros$ or $\Dones$, in which case
    \[
        \tvdist{f(\Ucal^m) - \Dcal} \le 1 \le O(\tvdist{f(\Ucal^m) - \Dcal_\Psi}).
    \]
    Similarly, we may assume $\tvdist{f(\Ucal^m) - \Dcal_\Psi} \le 1/3$.
    Applying \Cref{thm:middle_regime} or \Cref{thm:central_regime} depending on which regime $\Psi$ is in yields the result.
\end{proof}

\subsection{Tail Regime}\label{sec:middle_regime}

In this section, we deal with the regime where strings from $\Dcal_\Psi$ are spread out in the tail layers, but no weight is extremely close to the center, i.e., $\iota(\Psi)< n/2 - n^{2/3}$ or $\iota(\Psi) > n/2 + n^{2/3}$ is the Hamming weight in $\Psi$ closest to $n/2$.

\begin{theorem*}[\Cref{thm:middle_regime} Restated]
    Let $f \colon \bin^m \to \bin^n$ be a $d$-local function with $n$ sufficiently large (in terms of $d$). If  $\tvdist{f(\Ucal^m) - \Dcal_\Psi} \le 1/3$ for some $\Psi$ in the tail regime, then $\Dcal_\Psi \in \cbra{\Dzeros, \Dones, \Dzeroones}$.
\end{theorem*}

The main idea is that in this regime, the output must be concentrated on few weights. Thus, we can apply known error bounds on the individual weights, and use \Cref{lem:tvdist_after_conditioning} to combine the errors. We formalize this idea below.

\begin{claim}\label{clm:small_support}
    Let $\Psi \subseteq \cbra{0,1,\dots, n}$ be a non-empty set satisfying $|\iota(\Psi) - n/2| > n^{2/3}$. If $n$ is sufficiently large, then there exists a non-empty set $\bar\Psi \subseteq \Psi$ of size $|\bar\Psi| \le O(n^{1/3})$ such that
    $$
    \tvdist{\Dcal_\Psi-\Dcal_{\bar\Psi}}\le 1/2.
    $$
\end{claim}
\begin{proof}
    Let $\Psi = \{s_1, \ldots, s_k\}$, where the weights are ordered nonincreasing in their distance from $n/2$. Additionally, let $\ell = Cn^{1/3}$ for a sufficiently large constant $C$. If $k < \ell$, simply set $\bar\Psi = \Psi$. Otherwise we will show it suffices to set $\bar\Psi = \cbra{s_1, \ldots, s_\ell}$.

    Consider an arbitrary index $i \in [k]$, and assume without loss of generality that $s_i \le n/2$. Note that $|s_i - n/2| > n^{2/3}$ by assumption, so we in fact have $s_i < \frac{n}{2} - n^{2/3}$. Then,
    \begin{align*}
        \Pr_{x\sim \Dcal_\Psi}\sbra{|x| = s_i \bigm| |x| \not\in \cbra{s_1, s_2, \dots,s_{i-1}}} &\ge \frac{\binom{n}{s_i}}{\sum_{\substack{j\le s_i \\ j \ge n-s_i}} \binom nj} = \frac{\binom{n}{s_i}}{2\sum_{j\le s_i} \binom nj} \\
        &\ge \frac{\binom{n}{s_i}}{2\binom{n}{s_i}\frac{n-s_i+1}{n-2s_i+1}} \tag{by \Cref{fct:binom_tail_asym}} \\
        &\ge \frac{n - 2\pbra{\frac{n}{2} - n^{2/3}}}{2n} = \Theta(n^{-1/3}).
    \end{align*}
    Thus, 
    \begin{align*}
        \Pr_{x\sim \Dcal_\Psi}\sbra{|x| \not\in \cbra{s_1, \dots, s_\ell}} &= \prod_{i=1}^{\ell} \Pr_{x\sim \Dcal_\Psi}\sbra{|x| \ne s_i \bigm| |x| \not\in \cbra{s_1, s_2, \dots,s_{i-1}}} \\
        &\le \pbra{1- \Theta(n^{-1/3})}^\ell \le \exp\cbra{-\ell \cdot \Theta(n^{-1/3})} \le 1/2.
    \end{align*}
    In other words, setting $\bar\Psi = \cbra{s_1, \ldots, s_\ell}$ yields $\tvdist{\Dcal_\Psi - \Dcal_{\bar\Psi}} \le 1/2.$
\end{proof}

\begin{theorem}[Combination of {\cite[Theorem 1.2]{filmus2023sampling} and \cite[Theorem 5.10]{kane2024locality}}]\label{thm:single_hamming_slice}
    Let $1 \le k \le n-1$ be an integer, and let $f\colon \bin^m \to \bin^n$ be a $d$-local function. Then
    \[
        \tvdist{f(\Ucal^m) - \Dcal_{\cbra{k}}} \ge 1 - O_d(n^{-1/2}).
    \]
\end{theorem}

We now proceed to the proof of the section's main result.

\begin{proof}[Proof of \Cref{thm:middle_regime}]
    We will prove the contrapositive. Suppose $\Psi \not\subseteq \cbra{0,n}$. If $|\Psi| \le n^{1/3}$, let $\bar\Psi = \Psi$; otherwise,
    let $\bar\Psi \subseteq \Psi$ as guaranteed by \Cref{clm:small_support}.
    We further prune $\bar\Psi$ to ensure no individual weight is sampleable by defining $\Psi^\dag = \bar\Psi\setminus \cbra{0,n}$. Observe that our previous assumptions imply $\Psi^\dag$ is non-empty. Thus, this removal changes the distribution minimally, as the support of $\Dcal_{\bar\Psi}$ has size at least $n$, and we removed at most two elements. In particular,
    \[
        \tvdist{\Dcal_{\Psi^\dag} - \Dcal_{\bar\Psi}} \le \frac{2}{n}.
    \]
    
    For each $s \in \Psi^\dag$, we apply \Cref{thm:single_hamming_slice} to obtain $\tvdist{f(\Ucal^m)-\Dcal_{\cbra{s}}}\ge1-O_d(n^{-1/2})$. Since $\Dcal_{\Psi^\dag}$ is the convex combination of $\Dcal_{\cbra{s}}$ for $s\in\Psi^\dag$, \Cref{lem:tvdist_after_conditioning} implies
    $$
        \tvdist{f(\Ucal^m)-\Dcal_{\Psi^\dag}}\ge 1 - O_d\pbra{\frac{|\Psi^\dag|}{\sqrt{n}}} \ge 1 - O_d\pbra{n^{-1/6}}.
    $$
    We conclude by applying the triangle inequality to deduce 
    \begin{align*}
        \tvdist{f(\Ucal^m)-\Dcal_{\Psi}} &\ge \tvdist{f(\Ucal^m)-\Dcal_{\Psi^\dag}} - \tvdist{\Dcal_{\Psi^\dag} - \Dcal_{\bar\Psi}}  - \tvdist{\Dcal_{\bar\Psi} - \Dcal_{\Psi}} \\
        &\ge 1 - O_d\pbra{n^{-1/6}} - \frac{2}{n} - \frac{1}{2} > \frac{1}{3}. \qedhere
    \end{align*}
\end{proof}

\subsection{Central Regime}\label{sec:central_regime}

In this section, we handle the regime where some Hamming weight is very close to the center, i.e., $n/2 - n^{2/3} \le \iota(\Psi)\le n/2 + n^{2/3}$.  

\begin{theorem*}[\Cref{thm:central_regime} Restated]
    Let $f \colon \bin^m \to \bin^n$ be a $d$-local function with $n$ sufficiently large (in terms of $d$). If $\tvdist{f(\Ucal^m) - \Dcal_\Psi} \le \eps$ for some $\Psi$ in the central regime, then $\tvdist{f(\Ucal^m) - \Dcal} \le O(\eps)$ for some $\Dcal \in \cbra{\Deven, \Dodd, \Dall}$.
\end{theorem*}

A crucial structural result in our analysis of this regime is the following local limit theorem, which says that any locally sampleable distribution close to a roughly centered uniform symmetric distribution must also be near to a mixture of $\Deven$ and $\Dodd$. The proof is somewhat involved, so we defer it to its own section (\Cref{sec:LLT}) for clarity. 

\begin{theorem}\label{LLT Theorem}
Let $\delta>0$ and $f\colon\{0,1\}^m\rightarrow \{0,1\}^n$ be a $d$-local function with $n$ sufficiently large (in terms of $\delta$ and $d$). Let $\Psi \subseteq \{0,1,2,\ldots,n\}$ be a set containing some element $n(1/2\pm c(d,\delta))$ for some $c(d,\delta)>0$ a small enough function of $d$ and $\delta$. Then there exists a distribution $\Mcal$ which is a mixture of $\Deven$ and $\Dodd$ so that
$$
\tvdist{f(\Ucal^m) - \Mcal} \leq O(\tvdist{f(\Ucal^m) - \Dcal_\Psi}) + \delta.
$$
\end{theorem}

The following two results used in proving \Cref{LLT Theorem} will also be useful within this section. (Their proofs can likewise be found in \Cref{sec:LLT}.) The first allows us to reason about the distance between two distributions $A$ and $B$ over $\bin^n$ with $B$ symmetric by reasoning about their weight distributions and $A$'s symmetry. It is essentially a consequence of the observation that $A$ is far from $B$ if $A$ is either far from its own symmetrization or if its weight distribution noticeably differs from $B$'s. Recall $|A|$ denotes the distribution over the Hamming weight of strings $x\sim A$, and $A_\sym$ denotes the distribution resulting from randomly permuting strings $x\sim A$.

\begin{lemma}\label{lem:distance_to_sym}
Let $A$ and $B$ be two distributions on $\{0,1\}^n$ with $B$ symmetric. Then
$$
\tvdist{A-B} = \Theta(\tvdist{|A|-|B|} + \tvdist{A-A_\sym}).
$$
\end{lemma}

The second result can be viewed as a continuity property of the weight distribution of $f(\Ucal^m)$ in this section's regime.

\begin{proposition}\label{continuity prop}
Let $f\colon\{0,1\}^m\rightarrow \{0,1\}^n$ be a $d$-local function with $n$ sufficiently large (in terms of $d$). Let $\Psi \subseteq \{0,1,\ldots,n\}$ be a set containing an element $n(1/2\pm c(d))$ for some $c(d)>0$ a small enough function of $d$. Then the distribution $f(\Ucal^m)$ can be written as a mixture $aE + (1-a)X$ with $a=O(\tvdist{f(\Ucal^m) - \Dcal_\Psi})$ so that for any even $\Delta$ and $x\in\{0,1,\ldots,n\}$,
$$
\big|\Pr\sbra{|X| = x} - \Pr\sbra{|X| = x+\Delta} \big| = O_d\pbra{\frac{|\Delta|}{n}}.
$$
\end{proposition}

Here, one should view $E$ as the ``error'' part of $f(\Ucal^m)$ which is far from every roughly centered uniform symmetric distribution, while the remaining part $X$ has a smooth weight distribution (modulo parity constraints). These properties are ultimately inherited via the structure of the independent neighborhoods obtained after conditioning, as sketched in \Cref{sec:overview}. $E$ corresponds to neighborhoods that are far from uniform, while the smoothness of $X$ is a consequence of the anticoncentration properties of many roughly unbiased independent neighborhoods (formalized in \Cref{app:comp_llt}).

We proceed with an important corollary of \Cref{LLT Theorem}:
\begin{corollary}\label{close special cor}
Let $\delta>0$ and $f\colon\{0,1\}^m\rightarrow \{0,1\}^n$ be a $d$-local function with $n$ sufficiently large (in terms of $\delta$ and $d$). Let $\Psi \subseteq \{0,1,2,\ldots,n\}$ contain some element within $n^{2/3}$ of $n/2$. Then there exists a distribution $\Dcal \in \cbra{\Deven, \Dodd, \Dall}$ so that
$$
\tvdist{f(\Ucal^m) - \Dcal} \leq O(\tvdist{f(\Ucal^m) -\Dcal_\Psi}) + \delta.
$$
\end{corollary}

\begin{proof}
    For sufficiently large $n$, we may apply \Cref{LLT Theorem} to deduce that for some mixture $\Mcal = \eta \cdot \Deven + (1-\eta)\cdot \Dodd$, we have
    \begin{equation}\label{close special cor:eq:2}
        \tvdist{f(\Ucal^m) - \Mcal} \leq O(\tvdist{f(\Ucal^m) - \Dcal_\Psi}) + \frac{\delta}{3}.
    \end{equation}
    We will show a similar upper bound for $\tvdist{f(\Ucal^m) - \Dcal}$.
    This essentially follows from a number of applications of the triangle inequality. We have
        \begin{equation}\label{eq:lem:result_via_close_to_mixture}
        \tvdist{f(\Ucal^m) - \Dcal} \le \tvdist{f(\Ucal^m) - \Dcal_\Psi} + \tvdist{\Dcal_\Psi - \Dcal},
    \end{equation}
    so it remains to bound $\tvdist{\Dcal_\Psi - \Dcal}$. We first show $\tvdist{\Mcal - \Dcal} \le 2\tvdist{\Mcal - \Dcal_\Psi}$. 
    Let $S$ be a partition of $\{0,1\}^n$ into pairs of elements $(x_e, x_o)$ such that each pair contains one element $x_e$ of even weight and one element $x_o$ of odd weight. Then,
    \begin{align*}
        \tvdist{\Mcal - \Dcal} &= \frac{1}{2} \sum_{x\in \{0,1\}^n} |\Mcal(x) - \Dcal(x)| \\
        &= \frac{1}{2} \sum_{(x_e, x_o)\in S} |\Mcal(x_e) - \Dcal(x_e)| + |\Mcal(x_o) - \Dcal(x_o)| \\
        &= 2^{n-2} \pbra{\left|\frac{\eta}{2^{n-1}} - \Dcal(x_e)\right| + \left|\frac{1-\eta}{2^{n-1}} - \Dcal(x_o)\right|}.
    \end{align*}
    Breaking into cases, we find that
    \[
        \tvdist{\Mcal - \Dcal} =
        \begin{cases}
            1 - \eta & \text{ if } \Dcal = \Deven, \\
            \eta & \text{ if } \Dcal = \Dodd, \\
            \left|\eta - \frac{1}{2}\right| & \text{ if } \Dcal = \Dall.
        \end{cases}
    \]
    By choosing $\Dcal$ appropriately, we may assume $\tvdist{\Mcal - \Dcal} = \zeta \coloneqq \min\cbra{\eta, |\eta - 1/2|, 1-\eta}$. Similarly, 
    \begin{align*}
        \tvdist{\Mcal - \Dcal_\Psi} &= \frac{1}{2} \sum_{x\in \{0,1\}^n} |\Mcal(x) - \Dcal_\Psi(x)| \\
        &= \frac{1}{2} \sum_{(x_e, x_o)\in S} \left|\frac{\eta}{2^{n-1}} - \Dcal_\Psi(x_e)\right| + \left|\frac{1-\eta}{2^{n-1}} - \Dcal_\Psi(x_o)\right|.
    \end{align*}
    Since $\Dcal_\Psi$ is uniform, each pair $(x_e, x_o) \in S$ either satisfies $\Dcal_\Psi(x_e) = \Dcal_\Psi(x_o)$ or one of $\Dcal_\Psi(x_e), \Dcal_\Psi(x_o)$ is zero. Again breaking into cases,
    \[
        2^{n-1}\pbra{\left|\frac{\eta}{2^{n-1}} - \Dcal_\Psi(x_e)\right| + \left|\frac{1-\eta}{2^{n-1}} - \Dcal_\Psi(x_o)\right|} \ge \begin{cases}
            \eta & \text{ if } \Dcal_\Psi(x_e) = 0, \\
            1-\eta & \text{ if } \Dcal_\Psi(x_o) = 0, \\
            |2\eta-1| & \text{ if } \Dcal_\Psi(x_e) = \Dcal_\Psi(x_o).
        \end{cases} 
    \]
    In other words,
    \[
        \left|\frac{\eta}{2^{n-1}} - \Dcal_\Psi(x_e)\right| + \left|\frac{1-\eta}{2^{n-1}} - \Dcal_\Psi(x_o)\right| \ge \frac{\zeta}{2^{n-1}},
    \]
    so
    \begin{equation}\label{close special cor:eq:1}
        \tvdist{\Mcal - \Dcal_\Psi} \ge \frac{1}{2} \cdot 2^{n-1} \cdot \frac{\zeta}{2^{n-1}} = \frac{1}{2}\tvdist{\Mcal - \Dcal}.
    \end{equation}
    Therefore,
    \begin{align*}
        \tvdist{\Dcal_\Psi - \Dcal} &\le \tvdist{\Dcal_\Psi - \Mcal} + \tvdist{\Mcal - \Dcal} \\
        &\le 3\tvdist{\Mcal - \Dcal_\Psi} \tag{by \Cref{close special cor:eq:1}} \\
        &\le 3\pbra{\tvdist{\Mcal - f(\Ucal^m)} + \tvdist{f(\Ucal^m) - \Dcal_\Psi}}.
    \end{align*}
    Combining with \Cref{close special cor:eq:2} and \Cref{eq:lem:result_via_close_to_mixture} yields the result.
\end{proof}

Ideally, we would want to prove \Cref{thm:central_regime} by simply applying \Cref{close special cor} with $\delta = O(\tvdist{f(\Ucal^m) - \Dcal_\Psi})$. However, if $\delta$ is too small, $n$ may not be large enough (in terms of $\delta$ and $d$) to satisfy the assumption of \Cref{close special cor}. Thus, we can only deduce that \Cref{thm:central_regime} holds unless $\tvdist{f(\Ucal^m) - \Dcal}$ is less than a small constant. Hence, we will turn our attention to this special case. The two main consequences of such an assumption can be encapsulated in the following claims.

\begin{claim}\label{constant parity lem}
Let $f\colon\bin^m\to\bin^n$ be a $d$-local function. If $\tvdist{f(\Ucal^m) - \Dcal} < 2^{-d}$ for some $\Dcal \in \{\Deven,\Dodd\}$, then $\supp{f(\Ucal^m)}\subseteq\supp{\Dcal}$.
\end{claim}
\begin{proof}
Observe that we may write $|f(x)| \bmod 2$ as an $\mathbb{F}_2$-polynomial $p$ (of the input bits) of degree at most $d$. If $p$ is not constant, then it must take each possible value with probability at least $2^{-d}$, contradicting our distance assumption.
\end{proof}

We call a distribution over $\bin^n$ \emph{k-wise independent} if the projection onto any $k'\le k$ indices is uniformly distributed over $\bin^{k'}$.

\begin{claim}\label{moment matching lemma}
Let $f\colon\bin^m\to\bin^n$ be a $d$-local function. If $\tvdist{f(\Ucal^m) - \Dcal} < 2^{-kd}$ for some $\Dcal \in \{\Deven,\Dodd,\Dall\}$ and positive integer $k<n$, then $f(\Ucal^m)$ is a $k$-wise independent distribution.
\end{claim}

\begin{proof}
For any set $T \subseteq [n]$ of coordinates of size $k$, the data processing inequality implies 
\begin{equation}\label{eq:granularity_upper}
    2^{-kd} > \tvdist{f(\Ucal^m) - \Dcal} \ge \tvdist{f(\Ucal^m)[T] - \Dcal[T]}.
\end{equation}
Note that $\Dcal[T]$ is the uniform distribution on $\{0,1\}^T$. On the other hand, the outputs of $f$ on $T$ only depend on at most $kd$ input bits, so the probability that $f(\Ucal^m)[T]$ assigns to any given string must be a multiple of $2^{-kd}$. Thus, $\tvdist{f(\Ucal^m)[T] - \Dcal[T]}$ is also a multiple of $2^{-kd}$. Combining this observation with \Cref{eq:granularity_upper} yields $f(\Ucal^m)[T] = \Dcal[T]$, concluding the proof.
\end{proof}

Using these claims, we will show that the distribution on weights of our sampled distribution $|f(\Ucal^m)|$ is ``missing'' a noticeable amount of mass around $n/2$. For this, we record the observation that we can write $|f(\Ucal^m)|(x) = |\Dcal|(x) + \kappa A(x) - \kappa S(x)$ for some disjoint distributions $A,S$ and $\kappa \coloneqq \tvdist{|f(\Ucal^m)| - |\Dcal|} \le \tvdist{f(\Ucal^m) - \Dcal}$. 

\begin{proposition}\label{prop:mass_from_S}
Let $f\colon \bin^m \to \bin^n$ be a $d$-local function with $n$ sufficiently large (in terms of $d$). 
Suppose $\tvdist{f(\Ucal^m) - \Dcal} = \lambda$ for some $\Dcal \in \{\Deven,\Dodd,\Dall\}$ and $\lambda$ less than a sufficiently small constant (in terms of $d$). 
Moreover, let $A,S$ be distributions with disjoint support satisfying $|f(\Ucal^m)| = |\Dcal| + \kappa A - \kappa S$ for some $\kappa \le \lambda$. 
If $\tvdist{f(\Ucal^m) - \Dcal_\Psi} \le O(\lambda)$ for some sufficiently small implicit constant and $\Psi\subseteq \cbra{0,1,\dots, n}$, then
\begin{enumerate}
    \item $\kappa = \Omega(\lambda)$, and

    \item $S$ assigns at least a constant fraction of its mass to weights within $O(\sqrt{n})$ of $n/2$.
\end{enumerate} 
\end{proposition}

Note that if the assumption $\tvdist{f(\Ucal^m) - \Dcal_\Psi} \le O(\lambda)$ does not hold, there is nothing left to prove (i.e., the conclusion of \Cref{thm:central_regime} is already satisfied).

\begin{remark}
    In the following proof, we will want to apply hypercontractivity (\Cref{lem:hypercontractivity}) for expectations over any $\Dcal \in \cbra{\Deven, \Dodd, \Dall}$ (as opposed to over $\Ucal^n$). Observe, however, that each of these distributions is $(n-1)$-wise independent. Thus, $\E[p(\Ucal^n)] = \E[p(\Dcal)]$ for any polynomial $p$ of degree at most $n-1$. In particular, \Cref{lem:hypercontractivity} still holds when $q$ is even and $q\cdot\deg(p) < n$.

    Similarly, we will need to apply \Cref{fct:hoeffding} for $|\Dcal|$. In the case of $\Dcal = \Deven$, for example, we observe that
    \[
        \Pr_{x \sim |\Dcal|}[x = z] = \Pr_{x \sim |\Ucal^n|}[x = z \big| |x| \text{ is even}] \le \frac{\Pr_{x \sim |\Ucal^n|}[x = z]}{\Pr_{x \sim |\Ucal^n|}[|x| \text{ is even}]} = 2\Pr_{x \sim |\Ucal^n|}[x = z].
    \]
    Thus, we may use \Cref{fct:hoeffding} up to some small constant loss.
\end{remark}

\begin{proof}[Proof of \Cref{prop:mass_from_S}]
For clarity, let $W = |f(\Ucal^m)|$. To prove the first part of the claim, we note that by \Cref{lem:distance_to_sym},
\begin{align*}
    \lambda = \tvdist{f(\Ucal^m) - \Dcal} 
    &= \Theta\pbra{\tvdist{f(\Ucal^m) - f(\Ucal^m)_\sym} + \tvdist{W - |\Dcal|}} \\
    &= \Theta(\tvdist{f(\Ucal^m) - f(\Ucal^m)_\sym}+\kappa).
\end{align*}
This means that either $\kappa = \Omega(\lambda)$ or $\tvdist{f(\Ucal^m) - f(\Ucal^m)_\sym} = \Omega(\lambda)$. 
However, if the latter is true, then \Cref{lem:distance_to_sym} shows that $\tvdist{f(\Ucal^m) - \Dcal_\Psi} \ge \Omega(\lambda)$. 
This contradicts our assumption that $\tvdist{f(\Ucal^m) - \Dcal_\Psi} \le O(\lambda)$ for a sufficiently small implicit constant. 
Henceforth, we will assume 
\begin{equation}\label{eq:tvd_lb_W_minus_Dcal}
    \kappa = \Theta(\lambda).
\end{equation}

We now turn to showing $S$ assigns a constant fraction of its mass to weights close to $n/2$. First note that $\supp{W} \subseteq \supp{|\Dcal|}$ by \Cref{constant parity lem}. Therefore, we may assume that $\Psi \subseteq \supp{|\Dcal|}$ as well, since replacing $\Psi$ with $\Psi \cap \supp{|\Dcal|}$ will not increase the distance between $\Dcal_\Psi$ and $f(\Ucal^m)$.

In particular, one should view $|\Dcal_\Psi|$ as the result of redistributing mass from weights in $\supp{|\Dcal|}\setminus \Psi$ evenly over the elements in $\Psi$. Since $W$ is close to $|\Dcal_\Psi|$, we can roughly view $S$ as corresponding to the missing weights, and $A$ as corresponding to the small ``boost'' from the redistributed mass. A priori, it is difficult to reason about $S$'s distribution, as we have limited information about where the missing weights are. However, 
\begin{equation}\label{eq:DvsDpsi}
    \gamma \coloneqq \tvdist{\Dcal_\Psi - \Dcal} \le \tvdist{\Dcal_\Psi - f(\Ucal^m)} + \tvdist{f(\Ucal^m) - \Dcal} \le O(\lambda)
\end{equation}
is less than a small constant, so many weights around $n/2$ must be included in $\Psi$, which will receive further weight from $A$. Therefore, we will proceed by arguing $A$ assigns substantial mass to weights in a small interval around $n/2$, and then convert that conclusion to one about $S$.

We begin by observing that all but a small constant fraction of the support of $|\Dcal|$ in $I = [n/2 - \sqrt{n}/10,n/2 + \sqrt{n}/10]$ are in $\Psi$. Indeed by \Cref{fct:individual_binom} and \Cref{fct:entropy}, each element of $\supp{|\Dcal|} \cap (I \setminus \Psi)$ contributes at least
\[
    2^{-n}\binom{n}{\frac{n}{2} + \frac{\sqrt{n}}{10}} \ge \frac{2^{n\pbra{\Hcal\pbra{\frac{1}{2}-\frac{1}{10\sqrt{n}}} - 1}}}{\sqrt{2n\pbra{1 - \frac{1}{25n}}}} \ge \frac{1}{2\sqrt{n}}
\]
to the distance $\tvdist{\Dcal_\Psi - \Dcal}$, so by \Cref{eq:DvsDpsi} such elements can only consist of an $O(\lambda)$-fraction of $I$.

Next, note that
\[
    1 = \sum_{s \in \Psi} |\Dcal_\Psi|(s) = \gamma + \sum_{s \in \Psi}|\Dcal|(s),
\]
so $|\Dcal_\Psi|(s) = |\Dcal|(s)/(1-\gamma)$ for all $s\in\Psi$. This implies
\begin{align}
    \sum_{x\in I}\max\cbra{|\Dcal_\Psi|(x) - |\Dcal|(x), 0} &= \sum_{x \in I \cap \Psi} \frac{|\Dcal|(x)}{1-\gamma} - |\Dcal|(x) \notag \\
    &\ge \gamma\sum_{x \in I \cap \Psi}|\Dcal|(x) \notag \\
    &\ge \tvdist{\Dcal_\Psi - \Dcal}\cdot \Omega(1-\lambda) \notag \\
    &\ge \Omega(\tvdist{f(\Ucal^m) - \Dcal} - \tvdist{f(\Ucal^m) - \Dcal_\Psi}) \tag{by triangle inequality} \\
    &\ge \Omega(\lambda). \label{eq:max_diff_LB}
\end{align}
Now recall we originally assumed for some sufficiently small implicit constant that
\begin{align*}
    O(\lambda) &= \tvdist{f(\Ucal^m) - \Dcal_\Psi} \\
    &\ge \tvdist{W - |\Dcal_\Psi|} \tag{by data processing inequality} \\
    &\ge \sum_{x\in I}\max\cbra{|\Dcal_\Psi|(x) - W(x), 0} \\
    &\ge \sum_{x\in I}\max\cbra{|\Dcal_\Psi|(x) - |\Dcal|(x) - \kappa A(x), 0} \\
    &\ge \sum_{x\in I}\max\cbra{|\Dcal_\Psi|(x) - |\Dcal|(x), 0} - \kappa A(x) \\
    &\ge \Omega(\lambda) - \sum_{x\in I}\kappa A(x) \tag{by \Cref{eq:max_diff_LB}}.
\end{align*}
Therefore, $\sum_{x\in I}\kappa A(x) = \Omega(\lambda)$. Moreover, $A$ must assign constant mass to $I$ by \Cref{eq:tvd_lb_W_minus_Dcal}. 

Our strategy will be to leverage this fact to show $S$ must also assign a constant fraction of its mass to an interval of width $O(\sqrt{n})$ around $n/2$. Let $r$ be the largest odd integer so that $2r < \log(1/\lambda)/d$. Combining with \Cref{eq:tvd_lb_W_minus_Dcal}, we record the following equation for later convenience:
\begin{equation}\label{eq:1overlambdabound}
    \frac{1}{\kappa} = \Theta\pbra{\frac{1}{\lambda}} = \exp\pbra{O(dr)}.
\end{equation}
By \Cref{moment matching lemma}, we have that $W$ matches its first $2r$ moments with the binomial distribution $|\Ucal^n|$, or equivalently, with $|\Dcal|$. Thus, any polynomial $p$ of degree at most $2r$ satisfies $\E[p(W)] = \E[p(|\Dcal|)]$. In particular by the definition of $W$ and the fact that for any distribution $\Pcal$ and function $q$,  $\E_{x\sim \Pcal}[q(x)]$ is a linear functional on $\Pcal$, we have
\begin{equation}\label{eq:matching_moments_polys}
    \E[p(A)] = \E[p(S)].
\end{equation}

Suppose momentarily we were able to set $p$ to be the indicator for the interval $I$. Then the proposition 
would follow, since our earlier deduction about $A$'s distribution gives $\E[p(A)] = \Pr[A \in I] = \Omega(1)$, so \Cref{eq:matching_moments_polys} implies $\Pr[S \in I] = \Omega(1)$. While this scenario is of course unrealistic, as the indicator for $I$ will have an unaffordably high degree, the overall idea motivates our strategy of working with a carefully chosen polynomial $p$. 

We morally want to choose $p$ to approximate the indicator for $I$. However, we must be a bit discerning in our choice; polynomials blow-up in their tails, so any choice of $p$ will inevitably be a poor approximation over the entire domain. We can distill the desired behavior for our low-degree polynomial $p$ into the following three constraints:
\begin{enumerate}
    \item \label{itm:high_level_cheby_1} $p$ puts substantial mass on weights around $n/2$,

    \item \label{itm:high_level_cheby_2} $p$ puts minimal mass on weights further from the center, and

    \item \label{itm:high_level_cheby_3} $p$'s inevitable blow-up occurs as far into the tails as possible.
\end{enumerate}

Ultimately, we set $p(x) = (T_r(y/r)/y)^2$ where $y\coloneqq(x-n/2)/\sqrt{n}$ and $T_r$ is the Chebyshev polynomial (of the first kind) $T_r(\cos(\theta)) = \cos(r\theta)$. Intuitively, Chebyshev polynomials remain relatively flat on the interval $[-1,1]$, which allows us to (partially) achieve \Cref{itm:high_level_cheby_2}. The division by $r$ ``stretches'' the flat region to further guarantee \Cref{itm:high_level_cheby_2} and to address \Cref{itm:high_level_cheby_3} by pushing the blow-up further into the tails. Finally, the additional $1/y$ factor and the choice of the largest possible (affordable) degree satisfies \Cref{itm:high_level_cheby_1}. We formalize these properties below.

\begin{fact}\label{p facts}
We have that:
\begin{enumerate}
\item \label{p facts:1} $p(x)$ is a polynomial of degree at most $2r$.
\item \label{p facts:2} $p(x) \geq 1/2$ when $|y| \le 1/10.$
\item \label{p facts:3} $p(x) \leq \min(1,1/y^2)$ when $|y|\leq r$.
\item \label{p facts:4} $p(x) \leq O(y/r)^{2r}$ when $|y| \geq r$.
\item \label{p facts:5} $p(x) \geq 0$ for all $x$.
\end{enumerate}
\end{fact}
\begin{proof} ~
\begin{enumerate}
    \item Recall $T_r$ is a polynomial of degree $r$. Since $r$ is odd, $T_r$ is odd, and thus has a root at $y=0$. In particular, $T_r(y/r)/y$ is a polynomial of degree at most $r-1$.

    \item Let $\sin(\theta)=y/r$. Then $p(x) = |\sin(r\theta)/y|^2.$ Note that $|\theta| \le 1/(5r)$ when $|y| < 1/10$, so both $\sin(r\theta)/(r\theta)$ and $\sin(\theta)/\theta$ are in $[0.99,1]$. Thus, 
    $$
    p(x) = |\sin(r\theta)/(r\sin(\theta))|^2 = |\sin(r\theta)/(r\theta)|^2|\sin(\theta)/\theta|^{-2} \geq 1/2.
    $$

    \item Note that $|T_r(x)|\leq 1$ for $|x|\leq 1$ and so $p(x) \leq 1/y^2$. Finally, as above
    $$
    p(x) = |\sin(r\theta)/(r\sin(\theta))|^2 \leq 1.
    $$

    \item If $a=y/r$ we note that $a=(z+z^{-1})/2$ for some real $z$ with $1\leq |z| \leq 2|a|$. Setting $z = e^{i\theta}$, we can write $\cos(\theta) = (z+z^{-1})/2$, so $T_r(a) = T_r((z+z^{-1})/2) = (z^r+z^{-r})/2 \le (2|a|)^r.$ 
    
    \item Note that $p$ is a square. \qedhere
\end{enumerate}
\end{proof}

Recall $A$ must assign constant mass to $I = [n/2 - \sqrt{n}/10,n/2 + \sqrt{n}/10]$. Moreover, \Cref{p facts:2} and \Cref{p facts:5} imply $\E[p(A)] = \Omega(1).$ In particular, for some large enough constant $C$, we have that $\E[p(A)] > 1/C$. Therefore $\E[p(S)]$ is also at least $1/C$ by \Cref{eq:matching_moments_polys}.

We consider the contribution to $\E[p(S)]$ coming from three ranges: $|y| \leq C$, $C \leq |y| \leq r$, and $|y| \geq r$. 
(If $C \ge r$, we ignore the second and third ranges.) 
By \Cref{p facts:3}, the contributions from the first and second ranges are at most $\Pr[|S-n/2| \leq C\sqrt{n}]$ and $1/C^2$, respectively. 
For the third range, recall that $0 \le W(x) = |\Dcal|(x) - \kappa S$ for any $x \in \supp{S}$. 
Thus, $\kappa S$ must assign less mass to each point than $|\Dcal|$ does. 
Therefore this contribution to $\E[p(S)]$ is at most
\begin{align*}
    & (1/\kappa) \E_{x\sim |\Dcal|}[p(x)\mathbbm{1}\{|x-n/2| > r\sqrt{n}\}] \\
    \le \: & (1/\kappa) \E_{x\sim |\Dcal|}[p^2(x)\mathbbm{1}\{|x-n/2| > r\sqrt{n}\}]^{1/2} \Pr_{x\sim |\Dcal|}\sbra{|x-n/2| > r\sqrt{n}}^{1/2} \tag{by Cauchy-Schwarz} \\
    \le \: & (1/\kappa)\E_{x\sim |\Dcal|}\sbra{O\pbra{\frac{|x-n/2|/\sqrt{n}}{r}}^{4r}}^{1/2}\exp(-\Omega(r^2)) \tag{by \Cref{p facts:4} \& \Cref{fct:hoeffding}} \\
    \le \: & (1/\kappa)\cdot r^{-2r} \cdot O(r)^{r}\cdot \exp(-\Omega(r^2)) \tag{by \Cref{lem:hypercontractivity}} \\
    = \: & \exp(O(dr)-\Omega(r^2)) = \exp(-\Omega(r^2)) \tag{by \Cref{eq:1overlambdabound}},
\end{align*}
with the last line holding for $r$ large enough relative to $d$ (which is achievable by enforcing $\lambda$ sufficiently small).
Thus, 
\[
    \frac{1}{C} < \E[p(S)] \le \Pr[|S-n/2| \leq C\sqrt{n}] + \frac{1}{C^2} + \exp(-\Omega(r^2)),
\]
so $S$ must assign at least constant mass to the interval $[n/2-C\sqrt{n},n/2+C\sqrt{n}].$
\end{proof}

We can now prove the main result of this subsection: \Cref{thm:central_regime}. Recall we aim to show that $\tvdist{f(\Ucal^m) - \Dcal} \le O(\tvdist{f(\Ucal^m) - \Dcal})$ for $\Psi$ containing an element close to $n/2$. Our strategy will be to restrict to the small error case via our local limit theorem, and then use the preceding results to enact a pairing argument (as sketched in \Cref{sec:overview}). 
\begin{proof}
    
By \Cref{close special cor}, there exists a distribution $\Dcal \in \cbra{\Deven, \Dodd, \Dall}$ so that
$$
\tvdist{f(\Ucal^m) - \Dcal} \leq O(\tvdist{f(\Ucal^m) -\Dcal_\Psi}) + \delta
$$
for some sufficiently small $\delta > 0$.
The result immediately follows if $\tvdist{f(\Ucal^m) -\Dcal_\Psi} > \delta$, so we may assume $\tvdist{f(\Ucal^m) - \Dcal_\Psi} \le \delta$. In particular, we have $\lambda \coloneqq \tvdist{f(\Ucal^m) - \Dcal}$ is less than some sufficiently small constant. From here, we assume by contradiction that
\begin{equation}\label{eq:contradiction_assumpt}
    \tvdist{f(\Ucal^m) -\Dcal_\Psi} \le O(\lambda)
\end{equation}
for a sufficiently small implicit constant. 

As above, let $W \coloneqq |f(\Ucal^m)| = |\Dcal| + \kappa A - \kappa S$ for $\kappa \coloneqq \tvdist{W - |\Dcal|} \le \lambda$ and $A, S$ some distributions with disjoint support. 
Once again by \Cref{constant parity lem}, we have that $W$ is supported only on the elements in the support of $|\Dcal|$, and we assume without loss of generality that $\Psi \subseteq \supp{|\Dcal|}$. This implies
\begin{enumerate}
    \item\label{itm:supp_A} $A(x) = 0$ for $x \in \supp{|\Dcal|}$,
    \item\label{itm:supp_S} $S(x) = 0$ for $x \not\in \supp{|\Dcal|}$, and
    \item\label{itm:relative_weight_in_psi} $|\Dcal_\Psi|(x) \ge |\Dcal|(x)$ for $x \in \Psi$.
\end{enumerate}

Next, we apply \Cref{prop:mass_from_S} to deduce
\begin{equation}\label{eq:kappa_omega_lambda}
    \lambda \ge \kappa \ge \Omega(\lambda)
\end{equation}
and $S$ assigns at least a constant fraction of its mass to weights within an interval $I$ of width $O(\sqrt{n})$ around $n/2$. In particular, \Cref{itm:supp_S} gives
\begin{equation}\label{eq:S_has_mass_near_middle}
    \sum_{x\in I\cap \supp{|\Dcal|}} S(x) \ge \Omega(1).
\end{equation}
At a high level, we aim to show there exist many nearby pairs $x,y \in I \cap \supp{|\Dcal|}$ with $x \not\in \Psi$ and $y \in \Psi$. Summing over all pairs will produce the desired contradiction $\tvdist{f(\Ucal^m) -\Dcal_\Psi} = \Omega(\lambda)$.

We first show there are many $x\in I \cap \supp{|\Dcal|}$ with $x\not\in \Psi$. Note that $|\Dcal|(x) - W(x) = \kappa S(x)$ for $x\in \supp{|\Dcal|}$ by \Cref{itm:supp_A}. Thus,
\begin{align*}
    \sum_{x\in I\cap \supp{|\Dcal|} \cap \Psi} S(x) &\le  \sum_{x\in I\cap \supp{|\Dcal|} \cap \Psi} \frac{|\Dcal_\Psi|(x) - W(x)}{\kappa} \tag{by \Cref{itm:relative_weight_in_psi}} \\
    &\le \frac{2}{\kappa}\cdot \tvdist{|\Dcal_\Psi| - W} \\
    &\le \frac{2}{\kappa}\cdot \tvdist{\Dcal_\Psi - f(\Ucal^m)} \tag{by data processing inequality} \\
    &\le \frac{2}{\kappa}\cdot O(\lambda) = O(1). \tag{by \Cref{eq:contradiction_assumpt} \& \Cref{eq:kappa_omega_lambda}}
\end{align*}
for some sufficiently small implicit constant. Combining with \Cref{eq:S_has_mass_near_middle} yields
\[
    \sum_{x\in I\cap \supp{|\Dcal|} \cap \Psi^c} S(x) = \Omega(1).
\]
However for any $x \in I \cap \supp{|\Dcal|} \cap \Psi^c$, we have
\[
    S(x) = \frac{|\Dcal|(x)}{\kappa} = \Theta\pbra{\frac{1}{\kappa\sqrt{n}}} = \Theta\pbra{\frac{1}{\lambda\sqrt{n}}}.
\]
In conclusion, it must be the case that there are $\Theta(\lambda \sqrt{n})$ many $x\in I \cap \supp{|\Dcal|}$ with $x\not\in \Psi$.

We consider a procedure where we take these elements one at a time and pair them with the closest unpaired element of $I\cap \Psi$ with the same parity. As each element only has to avoid at most $O(\lambda\sqrt{n})$ elements that are either not in $\Psi$ or have already been paired, we end up with $\Theta(\lambda\sqrt{n})$ disjoint pairs of elements $(x_i,y_i)\in I \cap \supp{|\Dcal|}$ so that $x_i \not\in \Psi$, $y_i\in \Psi$, $|x_i-y_i| = O(\lambda\sqrt{n})$, and $x_i$ and $y_i$ have the same parity. Note for any such $y_i \in I \cap \Psi$, \Cref{fct:individual_binom} and \Cref{fct:entropy} imply
\begin{equation}\label{clm:ys_are_heavy}
    |\Dcal_\Psi|(y_i) \ge 2^{-n}\binom{n}{\frac{n}{2} + O(\sqrt{n})} \ge \frac{2^{n\pbra{\Hcal\pbra{\frac{1}{2}-O\pbra{n^{-1/2}}} - 1}}}{\sqrt{2n\pbra{1 - O\pbra{n^{-1}}}}} \ge \Omega\pbra{\frac{1}{\sqrt{n}}}.
\end{equation}

Now we apply \Cref{continuity prop} to write $W= aE+(1-a)X$ where $a=O(\tvdist{f(\Ucal^m) - \Dcal_\Psi})$ and 
\begin{equation}\label{eq:continuity Ws}
    \big|X(x_i) - X(y_i)\big| = O_d(\lambda/\sqrt{n})
\end{equation}
for each $i$. In particular, we have that
\begin{align*}
\tvdist{f(\Ucal^m) - \Dcal_\Psi} & \geq \tvdist{|f(\Ucal^m)| - |\Dcal_\Psi|} \tag{by data processing inequality} \\
& \geq \frac{1}{2}\sum_i \big|W(x_i) - |\Dcal_\Psi|(x_i)\big| + \big|W(y_i) - |\Dcal_\Psi|(y_i)\big|\\
& \geq \frac{1}{2}\sum_i \big||\Dcal_\Psi|(x_i)-|\Dcal_\Psi|(y_i)\big| - \big|W(x_i) - W(y_i)\big| \tag{by triangle inequality} \\
& \geq \frac{1}{2}\sum_i |\Dcal_\Psi|(y_i) - (1-a)\big|X(x_i) - X(y_i)\big| - a\big|E(x_i) - E(y_i)\big| \tag{since $x_i \not\in \Psi$} \\
& \geq \frac{1}{2}\sum_i \Omega\pbra{\frac{1}{\sqrt{n}}} - \big|X(x_i) - X(y_i)\big| + a\pbra{\big|X(x_i) - X(y_i)\big| - \big|E(x_i) - E(y_i)\big|} \tag{by \Cref{clm:ys_are_heavy}} \\
& \geq \frac{1}{2}\sum_i \Omega\pbra{\frac{1}{\sqrt{n}}} - O_d\pbra{\frac{\lambda}{\sqrt{n}}} - a\sbra{|X(x_i) - E(x_i)| + |X(y_i) - E(y_i)|} \tag{by \Cref{eq:continuity Ws} and triangle inequality} \\
& = \Theta(\lambda\sqrt{n})\cdot\Omega\pbra{\frac{1}{\sqrt{n}}} - a\sum_{i=0}^n W(i)\\
& = \Omega(\lambda) - O\pbra{\tvdist{f(\Ucal^m) - \Dcal_\Psi}}.
\end{align*}
That is, $\tvdist{f(\Ucal^m) - \Dcal_\Psi} = \Omega(\lambda)$, which contradicts our assumption that $\tvdist{f(\Ucal^m) - \Dcal_\Psi} = O(\lambda)$ for a sufficiently small implicit constant.
\end{proof}

\section{Local Limit Theorem}\label{sec:LLT}

In this section, we will prove that any low-depth sampleable distribution that is reasonably close (in total variation distance) to a roughly centered uniform symmetric distribution is also close to a mixture of $\Deven$ and $\Dodd$. More precisely, we show:
\begin{theorem*}[\Cref{LLT Theorem} Restated]
Let $\delta>0$ and $f\colon\{0,1\}^m\rightarrow \{0,1\}^n$ be a $d$-local function with $n$ sufficiently large (in terms of $\delta$ and $d$). Let $\Psi \subseteq \{0,1,2,\ldots,n\}$ be a set containing some element $n(1/2\pm c(d,\delta))$ for some $c(d,\delta)>0$ a small enough function of $d$ and $\delta$. Then there exists a distribution $\Mcal$ which is a mixture of $\Deven$ and $\Dodd$ so that
$$
\tvdist{f(\Ucal^m) - \Mcal} \leq O(\tvdist{f(\Ucal^m) - \Dcal_\Psi}) + \delta.
$$
\end{theorem*}

The proof of this result is in two steps. First, we will show (\Cref{Kol prop}) that the output weight distribution $|f(\Ucal^m)|$ is close in Kolmogorov distance to the binomial distribution, even when restricting the weight's parity. Then we will show (\Cref{continuity prop}) that the distribution on $|f(\Ucal^m)|$ satisfies a continuity property. Together these will imply that the distribution on the weights of $f(\Ucal^m)$ is comparable to those of $\Mcal$. To obtain the final result, we apply the following lemma.

\begin{lemma*}[\Cref{lem:distance_to_sym} Restated]
Let $A$ and $B$ be two distributions on $\{0,1\}^n$ with $B$ symmetric. Then
$$
\tvdist{A-B} = \Theta(\tvdist{A-A_\sym} + \tvdist{|A|-|B|}).
$$
\end{lemma*}
\begin{proof}
First, we prove the upper bound. By the triangle inequality,
$$
\tvdist{A-B} \leq \tvdist{A-A_\sym}+\tvdist{A_\sym - B} = \tvdist{A-A_\sym}+\tvdist{|A| - |B|}.
$$
For the lower bound, we note on the one hand by the data processing inequality that
\[
    \tvdist{A-B} \geq \tvdist{|A|-|B|}.
\]
On the other hand by the triangle inequality,
\[
    \tvdist{A-B} \geq \tvdist{A-A_\sym} - \tvdist{A_\sym - B} \geq \tvdist{A-A_\sym} - \tvdist{|A| - |B|}.
\]
Combining, we find
\begin{align*}
    3\tvdist{A-B} &\geq (\tvdist{A-A_\sym} - \tvdist{|A| - |B|}) + 2 \tvdist{|A| - |B|} \\
    &= \tvdist{A-A_\sym} + \tvdist{|A| - |B|}. \qedhere
\end{align*}
\end{proof}

\subsection{The Kolmogorov Bound}\label{Kol Section}

In this section, we prove the following result.
\begin{proposition}\label{Kol prop}
Let $\delta > 0$ and $f\colon\{0,1\}^m\rightarrow \{0,1\}^n$ be a $d$-local function with $n$ sufficiently large (in terms of $\delta$ and $d$). Let $\Psi \subseteq \{0,1,\ldots,n\}$ be a set containing some element $n(1/2 \pm c(d,\delta))$ for some $c(d,\delta)>0$ a small enough function of $d$ and $\delta$. Then the distribution $f(\Ucal^m)$ can be written as a mixture $aE + (1-a)X$ with $a=O(\tvdist{f(\Ucal^m) -\Dcal_\Psi})$ so that for some $\eta\in [0,1]$ and all $t\in\mathbb{R}$:
$$
\big|\Pr\sbra{|X| > t \textrm{ and } |X| \textrm{ is even}} - \eta \Pr\sbra{|\Ucal^n|>t}\big| = O(\delta),
$$
and
$$
\big|\Pr\sbra{|X| > t \textrm{ and } |X| \textrm{ is odd}} - (1-\eta) \Pr\sbra{|\Ucal^n|>t}\big| = O(\delta).
$$
\end{proposition}

We briefly note the necessity of $a$ in the above decomposition of $f(\Ucal^m)$. It is possible for $f(\Ucal^m)$ to have some small part that is far from all such $\Dcal_\Psi$ (where ``small'' depends on the distance between $f(\Ucal^m)$ and $\Dcal_\Psi$), but the remainder of the distribution must have weight close in Kolmogorov distance to the binomial distribution $|\Ucal^n|$, even accounting for parity.

We begin by proving a special case. Recall a distribution $\Dcal$ over $\bin^n$ is k-wise independent if the projection of $\Dcal$ onto any $k'\le k$ indices is uniformly distributed over $\bin^{k'}$. The following result handles the case where $f(\Ucal^m)$ is nearly $k$-wise independent and none of the input degrees are too large.

\begin{lemma}\label{k indep lem}
Let $k\ge 2$ be an integer, and let $f\colon\{0,1\}^m \rightarrow \{0,1\}^n$ be a $d$-local function with $n$ sufficiently large. Suppose that no input bit of $f$ affects more than $n/A$ output bits for some $A>0$. Suppose furthermore that there is a set $S\subseteq [n]$ of size $|S|\leq \sqrt{n/A}$ such that $f(\Ucal^m)$ is $k$-wise independent on $[n]\setminus S$.
Then there exists $\eta \in [0,1]$ so that for any $\delta \in (0,1/2)$ and $t \in \mathbb{R}$, we have 
$$
\big|\Pr\sbra{|f(\Ucal^m)| > t \textrm{ and } |f(\Ucal^m)| \textrm{ is even}} - \eta \Pr\sbra{|\Ucal^n|>t} \big|
$$
and
$$
\big|\Pr\sbra{|f(\Ucal^m)| > t \textrm{ and } |f(\Ucal^m)| \textrm{ is odd}} - (1-\eta) \Pr\sbra{|\Ucal^n|>t} \big|
$$
are at most
$$
O\pbra{\frac{\log(1/\delta)^{O(d)^d}}{\sqrt{A\delta}} + \frac{\log(k)}{\sqrt{k}} + \delta}.
$$
\end{lemma}
In particular, we note that as long as $k$ and $A$ are sufficiently large, we can make the error as small as we like.

\begin{proof}
The basic idea is to use the fact that $k$-wise independence fools threshold functions. In particular, the following result says that if the output coordinates of $f(\Ucal^m)$ are $k$-wise independent, then the probability that $|f(\Ucal^m)|$ (which is simply the sum of those outputs) is bigger than $t$ will approximate the same expression for the binomial distribution.
\begin{theorem}[\cite{diakonikolas2010bounded}]\label{thm:kwise_ind}
    Let $\Dcal$ be a $k$-wise independent distribution on $\bin^n$, and let $\Ucal^n$ be the uniform distribution over $\bin^n$. Then for any $w_1, \ldots, w_n,\theta \in \Rbb$, we have
    \[
        \left|\Pr_{x \sim \Dcal} [w_1x_1 + \cdots + w_n x_n \ge \theta] - \Pr_{x \sim \Ucal^n}[w_1x_1 + \cdots + w_n x_n \ge \theta]\right| \le O\pbra{\frac{\log(k)}{\sqrt{k}}}.
    \]
\end{theorem}
To deal with the parities, we observe that the parity of $|f(\Ucal^m)|$ is a degree at most $d$ polynomial over $\mathbb{F}_2$ in the input bits. Using the following consequence of \cite[Theorem 3.1]{chattopadhyay2020xor}, we have that randomizing a small number of input bits can be used to effectively re-randomize the parity of the output.

\begin{theorem}\label{parity randomization theorem}
Let $p$ be a degree-$d$ polynomial over $\mathbb{F}_2^n$ and $\delta \in (0,1/2)$. There exists a subset $R\subseteq [n]$ with $|R| \leq \log(1/\delta)^{O(d)^d}$ so that if we write $p(x) = p(x_{R^c}, x_R)$ where $x_R$ and $x_{R^c}$ are the coordinates in $R$ and not in $R$ respectively, then with probability at least $1-\delta$ over the choice of a random value of $x_{R^c}$ we have that
\begin{equation}\label{parity randomization theorem eqn}
    \left|\Pr_{x_R}\sbra{p(x_{R^c},x_R) = 1} - \Pr_x\sbra{p(x)=1}\right| < \delta.
\end{equation}
\end{theorem}

That is, if we randomly fix the coordinates in $R^c$, then with high probability re-randomizing the coordinates in $R$ will essentially re-randomize the output of $p$.

Now, let the parity of the output of $f(x)$ be given by the $\mathbb{F}_2$-polynomial $p(x)$. Applying \Cref{parity randomization theorem} gives a set $R$ of at most $\log(1/\delta)^{O(d)^d}$ input bits. Let $B\subseteq [n]$ be the set of output bits affected by these input bits along with the output bits in $S$.
Note that by our assumptions, $|B| \leq n \log(1/\delta)^{O(d)^d}/A$.

We set $\eta$ to be the probability that $p(x) = 0$ over $x \sim \Ucal^m$ and will show that
\[
    \Pr\sbra{|f(\Ucal^m)| > t \textrm{ and } |f(\Ucal^m)| \textrm{ is even}} \\  
    \geq \eta \Pr\sbra{|\Ucal^n|>t} - O\pbra{\frac{\log(1/\delta)^{O(d)^d}}{\sqrt{A\delta}} + \frac{\log(k)}{\sqrt{k}} + \delta}.
\]
Observe that the case of odd $|f(\Ucal^m)|$ is almost identical, while the upper bounds follow from considering the complement distribution $1^n - f(\Ucal^m)$.

Before proceeding further, we define several variables and events for the sake of future clarity. Let $C\coloneqq |B|/2 - \sqrt{|B|/\delta}-|S|$. Let \textsf{EVEN} be the event that $|f(\Ucal^m)|$ is even (and similarly for \textsf{ODD}). Additionally, let \textsf{BIG} be the event that $|f(\Ucal^m)[B^c]| > t-C$. Finally, let \textsf{GOOD} be the event that $x_{R^c}$ satisfies \Cref{parity randomization theorem eqn} (and \textsf{BAD} be the complement event). We have
\begin{align*}
    \Pr\sbra{|f(\Ucal^m)| > t \textrm{ and } \textsf{EVEN}} &\ge \Pr\sbra{\textsf{BIG} \textrm{ and } |f(\Ucal^m)[B]| \ge C \textrm{ and } \textsf{EVEN}} \\
    &= \Pr\sbra{\textsf{BIG} \textrm{ and } \textsf{EVEN}} - \Pr\sbra{\textsf{BIG} \textrm{ and } |f(\Ucal^m)[B]| < C \textrm{ and } \textsf{EVEN}} \\
    &\ge \Pr\sbra{\textsf{BIG}}\cdot \Pr\sbra{\textsf{EVEN} \bigm| \textsf{BIG}} - \Pr\sbra{|f(\Ucal^m)[B]| < C}.
\end{align*}
Since the coordinates of $f(\Ucal^m)$ not in $B$ are $k$-wise independent by assumption, \Cref{thm:kwise_ind} implies
\begin{align}
    \Pr\sbra{\textsf{BIG}} &\ge \Pr\sbra{|\Ucal^{n-|B|}| > t-C} - O\pbra{\frac{\log(k)}{\sqrt{k}}} \notag \\ 
    &= \Pr\sbra{|\Ucal^{n}| > t-C+|\Ucal^{|B|}|} - O\pbra{\frac{\log(k)}{\sqrt{k}}} \notag \\
    &\ge \Pr\sbra{|\Ucal^{n}| > t} - \Pr\sbra{t \le |\Ucal^{n}| \le t-C+|\Ucal^{|B|}|} - O\pbra{\frac{\log(k)}{\sqrt{k}}} \notag \\
    &\ge \Pr\sbra{|\Ucal^{n}| > t} - \max_x |\Ucal^{n}|(x) \cdot \E\sbra{\big||\Ucal^{|B|}|-C\big|} - O\pbra{\frac{\log(k)}{\sqrt{k}}} \notag \\
    &\ge \Pr\sbra{|\Ucal^{n}| > t} - O\pbra{\frac{\big|C - \frac{|B|}{2}\big| + \sqrt{|B|}}{\sqrt{n}} + \frac{\log(k)}{\sqrt{k}}}. \label{eq:cdf_f_Bc}
\end{align}

To lower bound the conditional probability $\Pr\sbra{\textsf{EVEN} \bigm| \textsf{BIG}}$, it will be slightly more convenient to upper bound $\Pr\sbra{\textsf{ODD} \bigm| \textsf{BIG}}$. For this, observe that the event $\textsf{BIG}$ does not depend on any of the input bits in $R$. Hence by \Cref{parity randomization theorem}, for all but a $\delta$-fraction of the settings of the bits in $R^c$, the probability over the bits in $R$ that $|f(\Ucal^m)|$ is odd is at most $(1-\eta) + \delta$. Therefore, 
\begin{align*}
    \Pr\sbra{\textsf{ODD} \bigm| \textsf{BIG}} &= \Pr\sbra{\textsf{GOOD} \bigm| \textsf{BIG}}\cdot \Pr\sbra{\textsf{ODD} \bigm| \textsf{GOOD} \text{ and } \textsf{BIG}} \\
    & \ \ \: + \Pr\sbra{\textsf{BAD} \bigm| \textsf{BIG}}\cdot \Pr\sbra{\textsf{ODD} \bigm| \textsf{BAD} \text{ and } \textsf{BIG}} \\
    &\le \Pr\sbra{\textsf{ODD} \bigm| \textsf{GOOD}} + \Pr\sbra{\textsf{BAD} \bigm| \textsf{BIG}} \le (1-\eta) + \delta + \frac{\delta}{\Pr\sbra{\textsf{BIG}}}.
\end{align*}
Combining with \Cref{eq:cdf_f_Bc}, we have that
\begin{align*}
    \Pr\sbra{\textsf{BIG}}\Pr\sbra{\textsf{EVEN} \bigm| \textsf{BIG}} &\ge \eta \Pr\sbra{|\Ucal^n| > t} - O\pbra{\frac{\big|C - \frac{|B|}{2}\big| + \sqrt{|B|}}{\sqrt{n}} + \frac{\log(k)}{\sqrt{k}} + \delta} \\
    &\ge \eta \Pr\sbra{|\Ucal^n| > t} - O\pbra{\frac{\log(1/\delta)^{O(d)^d}}{\sqrt{A\delta}} + \frac{\log(k)}{\sqrt{k}} + \delta}.
\end{align*}
To complete our proof, it remains to bound $\Pr\sbra{|f(\Ucal^m)[B]| < C}$. We have
\begin{align*}
    \Pr\sbra{|f(\Ucal^m)[B]| < C} &\le \Pr\sbra{|f(\Ucal^m)[B\setminus S]| < |B|/2 - \sqrt{|B|/\delta}-|S|} \\
    &\le \Pr\sbra{|f(\Ucal^m)[B\setminus S]| < |B\setminus S|/2 - \sqrt{|B \setminus S|/\delta}} < \delta,
\end{align*}
where the final inequality follows from Chebyshev's inequality and the observation that the outputs in $B\setminus S$ are $2$-wise independent.
\end{proof}

We note that the degree bound in \Cref{k indep lem} can be achieved by restricting high degree inputs, which $d$-locality guarantees relatively few of.
To reduce to the case where most output coordinates are $k$-wise independent, we use the following lemma.

\begin{lemma}\label{tree lemma}
Let $f\colon\{0,1\}^m \rightarrow \{0,1\}^n$ be a $d$-local function, $k$ be a positive integer, and $\delta,\kappa \in (0,1]$. Additionally, let $\Psi \subseteq \{0,1,\ldots,n\}$ contain some element within $n/c$ of $n/2$ for some $c$ sufficiently large given the values of $d,k,\delta,$ and $\kappa$. Then there exists a partition of $\{0,1\}^m$ into subcubes of three types so that:
\begin{enumerate}
\item For each subcube $C$ of Type-I there is a set of $O_{d,k,\delta,\kappa}(1)$ coordinates in $[n]$ so that the other coordinates of $f(\Ucal(C))$ are $k$-wise independent. 

\item Letting $W$ be the union of subcubes of Type-II, we have that $\tvdist{f(\Ucal(W)) - \Dcal_\Psi} > 1-\kappa$.

\item The total probability mass of all subcubes of Type-III is at most $\delta$.
\end{enumerate}
\end{lemma}

We will later apply \Cref{tree lemma} to each subcube generated by conditioning on high degree input coordinates. This partitions the input space $\bin^m$ into subcubes, where almost all the mass is on subcubes $C$ such that either $f(\Ucal(C))$ is nearly $k$-wise independent with bounded degree inputs (in which case \Cref{k indep lem} is applicable) or whose union $W$ has $f(\Ucal(W))$ far from $\Dcal_\Psi$.

Our proof will require the following two consequences of hypercontractivity to analyze bounded-degree polynomials. The first follows from combining \Cref{lem:hypercontractivity}  with Markov's inequality.
Recall for $f\colon \cbra{-1,1}^n \to \Rbb$, we define the norm $\|f\|_2 = \pbra{\E_{X\sim \cbra{-1,1}^n}\sbra{|f(X)|^2}}^{1/2}$.\footnote{We originally defined the norm for functions on the domain $\bin^n$, but it will be more convenient in the subsequent proof to define it for the domain $\binpm^n$.}

\begin{lemma}[See, e.g., {\cite[Corollary 26]{kane2017structure}}]\label{lem:poly_anticoncentration}
    Let $K > 0$. If $p \colon \binpm^n \to \binpm$ is a degree-$d$ polynomial, then
    \[
        \Pr_{X\sim \cbra{-1,1}^n}\sbra{|p(X)| > K\cdot \|p\|_2} = O\pbra{2^{-(K/2)^{2/d}}}.
    \]
\end{lemma}
The second is a consequence of \Cref{lem:hypercontractivity} and the Paley-Zygmund inequality.
\begin{lemma}[See, e.g., {\cite[Theorem 2.4]{kabanets2017polynomial}}]\label{lem:weak_anticoncentration}
    If $p\colon\binpm^n \to \binpm$ is a degree-$d$ multilinear polynomial, then
    \[
        \Pr_{X\sim \cbra{-1,1}^n}\sbra{|p(X)| \ge (1/2) \cdot \|p\|_2} \ge (1/2)\cdot 9^{-d}.
    \]
\end{lemma}

\begin{proof}[Proof of \Cref{tree lemma}]
It will be convenient to view $f$ as a function $\{\pm 1\}^m\rightarrow \{\pm 1\}^n$ rather than $\{0,1\}^m\rightarrow \{0,1\}^n$ (as well as $\Ucal^n$ as the uniform distribution over $\cbra{\pm 1}^n$ and similarly for $\Dcal_\Psi$). As this change does not affect any of our claims, we will assume it throughout the following. We also assume throughout that $n$ is at least a sufficiently large function of $d,k,\delta,\kappa$ as otherwise, we can simply assign $\{\pm 1\}^m$ to be a single cube of Type-I, noting that there are at most $n$ coordinates which are not $k$-wise independent.

We will show the lemma statement holds via induction on $\left\lceil \log(1/\delta) / \log(1/(1-2^{-Bkd})) \right\rceil$ for some sufficiently large constant $B$ and all $d,k,\delta,\kappa$. For the base case $\left\lceil \log(1/\delta) / \log(1/(1-2^{-Bkd})) \right\rceil = 0$, it must be that $\delta=1$, so we declare all of $\{\pm 1\}^m$ to be a cube of Type-III. For the inductive step, we assume that our result holds for all $d',k',\delta',\kappa$ satisfying 
\[
    \left\lceil \frac{\log\pbra{\frac{1}{\delta'}}}{ \log\pbra{\frac{1}{1-2^{-Bk'd'}}}} \right\rceil < \left\lceil \frac{\log\pbra{\frac{1}{\delta}}}{ \log\pbra{\frac{1}{1-2^{-Bkd}}}} \right\rceil.
\]
Our basic strategy will be an iterative decomposition of $\{\pm 1\}^m$.

Consider all of the non-constant monomials of degree at most $k$ in the output bits whose expectations when evaluated on $f(\Ucal^m)$ are non-zero.
Take $N= 2^{6kd}\log^k(L/\kappa)$ for $L$ some sufficiently large constant. If there are fewer than $N$ such monomials, we can declare all of $\{\pm 1\}^m$ a single subcube of Type-I, noting that other than the at most $Nk$ coordinates involved in these monomials, the outputs of $f(\Ucal^m)$ are $k$-wise independent. Otherwise let $p$ be the sum of $N$ of these monomials times the sign of their individual expectations.

Observe that each of these monomials is a function of at most $kd$ input bits, so its expectation is a multiple of $2^{-kd}$. In particular, this implies that $\E[p(f(\Ucal^m))] \geq N 2^{-kd}$. On the other hand, since it has no constant term, $\E[p(\Ucal^n)]=0$ and $\Var(p(\Ucal^n)) = N$. Hence \Cref{lem:poly_anticoncentration} implies a concentration bound:
$$
\Pr\sbra{|p(\Ucal^n)| > N 2^{-kd-1}} = 
\Pr\sbra{|p(\Ucal^n)| > \sqrt{N} 2^{-kd-1}\cdot \|p\|_2}
\le O\pbra{2^{-(N2^{-2kd-4})^{1/k}}} < \kappa/10,
$$
where the final inequality follows from choosing $L$ (and thus $N$) sufficiently large. Additionally, we claim the restriction of $\Dcal_{\Psi}$ to the at most $Nk$ bits that $p$ depends on is within a factor of two of uniform. Indeed, let $Z \subseteq [n]$ be the set of output bits affecting $p$. 
\begin{claim}\label{clm:restrict_factor_of_two}
    For all $x \in \binpm^Z$ and $n$ sufficiently large, we have $\Dcal_\Psi[Z](x) \le 2\cdot \Ucal^Z(x)$.
\end{claim}
For clarity, we will finish the remainder of the proof before proving \Cref{clm:restrict_factor_of_two}. 
The claim implies
\begin{equation}\label{eq:tree lemma concentration}
    \Pr\sbra{|p(\Dcal_\Psi)| > N 2^{-kd-1}} < \kappa/5.
\end{equation}
Moreover, take the at most $Nkd$ input coordinates that affect the value of $p$. We split $\{\pm 1\}^m$ into subcubes based on all possible conditionings of these coordinates. We claim that at least a decent fraction of them can be assigned as Type-II. In particular, note that $p(f(x))$ is some polynomial $q(x)$ of degree at most $kd$ satisfying $\E[q(\Ucal^m)] \geq N 2^{-kd}$. Therefore by \Cref{lem:weak_anticoncentration}, with probability at least $2^{-Bkd}$ over this conditioning (which completely determines the value of $q$), the resulting subcube $C$ satisfies
$$
|q(C)| \geq \|q(\Ucal^m)\|_2 / 2 \ge \E[q(\Ucal^m)]/2 \geq N2^{-kd-1},
$$
where the second inequality follows from Jensen's inequality.

Hence, for at least a $2^{-Bkd}$-fraction of these subcubes we have that $|q| \geq N2^{-kd-1}.$ In particular, \Cref{eq:tree lemma concentration} implies the uniform distribution over the union of these subcubes is $(1-\kappa/5)$-far from $\Dcal_\Psi$. We declare all of these subcubes to be Type-II, and use our inductive hypothesis on each of the remaining $M \leq 2^{Nkd}$ subcubes with parameters $\kappa' = \kappa/(5M)$ and $\delta' = \delta/(1-2^{-Bkd})$, noting that
\[
    \left\lceil \frac{\log\pbra{\frac{1}{\delta'}}}{ \log\pbra{\frac{1}{1-2^{-Bkd}}}} \right\rceil = \left\lceil \frac{\log\pbra{\frac{1}{\delta}} + \log\pbra{1-2^{-Bkd}}}{ \log\pbra{\frac{1}{1-2^{-Bkd}}}} \right\rceil = \left\lceil \frac{\log\pbra{\frac{1}{\delta}}}{ \log\pbra{\frac{1}{1-2^{-Bkd}}}} \right\rceil - 1.
\]
(If $M = 0$, we have already obtained our desired partition.) For each subcube $i$ for $i = 1,2,\dots, M$, this gives a partition into three types such that

\begin{enumerate}
\item For each subcube $C$ of Type-I there is a set of $O_{d,k,\delta,\kappa}(1)$ coordinates in $[n]$ so that the other coordinates of $f(\Ucal(C))$ are $k$-wise independent.

\item Letting $W_i$ be the union of subcubes of Type-II, we have that $\tvdist{f(\Ucal(W_i)) - \Dcal_\Psi} > 1-\frac{\kappa}{5M}$.

\item The total probability mass of all subcubes of Type-III is at most $\delta'$.
\end{enumerate}

It remains to combine these partitions together, along with the union $W'$ of Type-II subcubes $C$ with $|q(C)| \ge N2^{-kd-1}$. The Type-I subcubes all remain Type-I during this combination. For Type-II subcubes, we have already shown that $\tvdist{f(\Ucal(W')) - \Dcal_\Psi} \ge 1 - \frac{\kappa}{5}$. Appealing to \Cref{lem:tvdist_after_conditioning}, we find that the union over all Type-II subcubes satisfies
\[
    \tvdist{f(\Ucal(W_1 \cup \cdots \cup W_M \cup W')) - \Dcal_\Psi} \ge 1 - (M+2)\cdot \frac{\kappa}{5M} \ge 1 - \kappa.
\] 
We conclude by calculating that the total probability mass of all subcubes of Type-III is at most $(1-2^{-Bkd})\delta' = \delta$, as desired.
\end{proof}

We now prove \Cref{clm:restrict_factor_of_two}, showing that $\Dcal_\Psi[Z](x) \le 2\cdot \Ucal^Z(x)$ for all $x \in \binpm^Z$. This allowed us to reason about the concentration of $|p(\Dcal_\Psi)|$ using information about the concentration of $|p(\Ucal)|$.

\begin{proof}[Proof of \Cref{clm:restrict_factor_of_two}]
 For this proof, it will be more convenient to revert back to working with $\bin$ rather than $\binpm$. 
 We may restrict our attention to the weights $\bar\Psi = \cbra{s \in \Psi : |s - n/2| \le 2n/c}$ at a small cost. Indeed, \Cref{fct:binom_tail_asym} implies
\begin{equation}\label{eq:clm:central_regime_truncate_1}
    A\coloneqq\sum_{s\in\Psi : |s - n/2| > \frac{2n}{c}}\binom ns\le2\cdot\sum_{0\le s<\frac{n}{2} - \frac{2n}{c}}\binom ns\le2\cdot2^{n\cdot\Hcal\pbra{\frac12 - \frac2c}}.
\end{equation}
Then by \Cref{fct:individual_binom}, we have
\begin{equation}\label{eq:clm:central_regime_truncate_2}
    B\coloneqq\binom n{\iota(\Psi)}\ge\binom n{\frac{n}{2} - \frac{n}{c}}\ge\frac{2^{n\cdot\Hcal\pbra{\frac12 - \frac1c}}}{\sqrt{8n\pbra{\frac12 - \frac1c}^2}}.
\end{equation}
Thus,
\begin{align*}
    \Pr_{s\sim \Dcal_\Psi}\sbra{\left|s-\frac{n}{2}\right| > \frac{2n}{c}} &= \frac A{|\supp{\Dcal_\Psi}|}\le\frac AB \le \sqrt{8n\pbra{\frac12 - \frac1c}^2}\cdot 2^{n\pbra{\Hcal\pbra{\frac12 - \frac2c} - \Hcal\pbra{\frac12 - \frac1c}}}.
\end{align*}
By \Cref{fct:entropy}, we have
\begin{align*}
\Hcal\pbra{\frac12 - \frac2c} - \Hcal\pbra{\frac12 - \frac1c}
&=\frac1{2\ln(2)}\sum_{m\ge1}\frac{\pbra{2/c}^{2m}-\pbra{4/c}^{2m}}{m\cdot(2m-1)} \\
&\le \frac{(2/c)^2 - (4/c)^2}{2\ln(2)} \le -\frac{1}{c^2}.
\end{align*}
Hence
\[
    \Pr_{s\sim \Dcal_\Psi}\sbra{\left|s-\frac{n}{2}\right| > \frac{2n}{c}} \le \sqrt{8n\pbra{\frac12 - \frac1c}^2} \cdot 2^{-n/c^2}.
\]
For sufficiently large $c$ and $n$, this implies 
    \[
        \Dcal_{\Psi}[Z](x) \le \Pr_{s \sim \Dcal_{\Psi}}[|s - n/2| > 2n/c] + \Dcal_{\bar\Psi}[Z](x) \le \frac{1}{2}\cdot \Ucal^Z(x) + \Dcal_{\bar\Psi}[Z](x),
    \]
    so it remains to show $\Dcal_{\bar\Psi}[Z](x) \le (3/2)\cdot \Ucal^Z(x)$.
    Assume without loss of generality $Z = [z]$. Then for all $x \in \binpm^z$ and $s \in \bar\Psi$, if we let $y \sim \Dcal_{\{s\}}[Z]$, we have
    \[
        \Dcal_{\{s\}}[Z](x) = \Pr[y_1 = x_1] \cdot \Pr\sbra{y_2 = x_2 \mid y_1 = x_1} \cdots \Pr\sbra{y_z = x_z \mid y_i = x_i \text{ for all } i \in [z-1]}.
    \]
    We will show each factor is at most $\frac{1}{2} + \frac{\lambda}{z}$ for some sufficiently small constant $\lambda>0$, so that
    \[
        \Dcal_{\bar\Psi}[Z](x) \le \max_{s \in \bar\Psi}\Dcal_{\{s\}}[Z](x) \le \pbra{\frac{1}{2} + \frac{\lambda}{z}}^z = 2^{-z}\cdot \pbra{1 + \frac{2\lambda}{z}}^z \le \Ucal^Z(x) \cdot e^{2\lambda} \le \frac{3}{2} \cdot \Ucal^Z(x).
    \]
    Consider the $i$-th factor in the case of $x_i = 1$. (The case of $x_i=0$ is similar.) Then
    \begin{align*}
        \Pr\sbra{y_i = 1 \mid y_j = x_j \text{ for all } j \in [i-1]} &= \frac{s - \sum_{j < i} \mathbbm{1}(x_j = 1)}{n-(i-1)} \le \frac{\frac{n}{2}+\frac{2n}{c}}{n-z} \le \frac{1}{2} + \frac{\lambda}{z}
    \end{align*}
    for $c,n$ sufficiently large.
\end{proof}

We are now ready to prove \Cref{Kol prop}. Recall our goal is to show $f(\Ucal^m)$ can be written as a mixture $aE + (1-a)X$ with $a = O(\tvdist{f(\Ucal^m) - \Dcal_\Psi})$, where $|X|$ is $O(\delta)$-close in Kolmogorov distance to the binomial distribution $|\Ucal^n|$, even accounting for parity. Our strategy will be to use the previously proven inductive decomposition lemma to partition the input space into structured subcubes. Notably, one type of resulting subcube is amiable to the application of our Kolmogorov distance / parity result (\Cref{k indep lem}).
\begin{proof}[Proof of \Cref{Kol prop}]
We will show that for any $\delta>0$, we can achieve error $O(\delta)$ provided $n$ is big enough. We assume $\delta \le 1/2$
as otherwise we can take $a=0$, $X=f(\Ucal^m)$, and $E$ arbitrary. 

Let $k$ be an integer at least $\log^2(1/\delta)/\delta^2$ and let $A = (1/\delta)^3 \log(1/\delta)^{\Omega(d)^d}$ sufficiently large.
We begin by conditioning on every input coordinate which affects more than $n/A$ many output coordinates. Note that since the sum of the degrees of our coordinates is at most $nd$, the number $M$ of these coordinates is at most $Ad$. This leaves us with a partition of $\{0,1\}^m$ into $2^M$ subcubes. On each of these subcubes we apply \Cref{tree lemma} with parameters $d$, $k$, $\delta$, and $\kappa = \delta 2^{-M}/10$. This yields a partition of all of $\{0,1\}^m$ into subcubes of types I, II, and III so that:
\begin{enumerate}
\item\label{itm:Kol prop_1} For each subcube $C$ of Type-I there is a set of $O_{d,k,\delta}(1)$ coordinates in $[n]$ so that the other coordinates of $f(\Ucal(C))$ are $k$-wise independent. 
\item\label{itm:Kol prop_2} Let $W$ be the union of subcubes of Type-II. By \Cref{lem:tvdist_after_conditioning} and the choice of $\kappa$, we have 
\[
    \tvdist{f(\Ucal(W)) - \Dcal_\Psi} > 1 - (2^M + 1)\kappa \ge 1-\frac{\delta}{2}.
\]
\item\label{itm:Kol prop_3} The total probability mass of all subcubes of Type-III is at most $\delta$.
\end{enumerate}
Recall that for each subcube $C$ of Type-I, \Cref{k indep lem} guarantees some $\eta_C \in [0,1]$ such that as long as $n$ is sufficiently large, we have
$$
\big|\Pr\sbra{|f(\Ucal(C))| > t \textrm{ and } |f(\Ucal(C))| \textrm{ is even}} - \eta_C \Pr\sbra{|\Ucal^n|>t} \big|
$$
and
$$
\big|\Pr\sbra{|f(\Ucal(C))| > t \textrm{ and } |f(\Ucal(C))| \textrm{ is odd}} - (1-\eta_C) \Pr\sbra{|\Ucal^n|>t} \big|
$$
are at most
$$
O\pbra{\frac{\log(1/\delta)^{O(d)^d}}{\sqrt{A\delta}}+\frac{\log(k)}{\sqrt{k}} + \delta} = O(\delta).
$$
Letting $U$ be the union of all Type-I cubes and $\eta$ be the appropriate weighted average of the $\eta_C$'s, we see that
$$
\big|\Pr\sbra{|f(\Ucal(U))| > t \textrm{ and } |f(\Ucal(U))| \textrm{ is even}} - \eta \Pr\sbra{|\Ucal^n|>t} \big| = O(\delta),
$$
and
$$
\big|\Pr\sbra{|f(\Ucal(U))| > t \textrm{ and } |f(\Ucal(U))| \textrm{ is odd}} - (1-\eta) \Pr\sbra{|\Ucal^n|>t} \big| = O(\delta).
$$

Furthermore, let $V$ be the union of $U$ and all the Type-III cubes, as well as the Type-II cubes if their total mass is at most $4\delta$. We note that the same inequality will hold for $f(\Ucal(V))$, which we set equal to $X$. If the total mass of the subcubes of Type-II is less than $4\delta$, we set $a=0$ and are done. Otherwise, we set $a$ to be the total mass of these Type-II subcubes and let $E=f(\Ucal(W))$, where recall $W$ is the union over the subcubes of Type-II. However, we see by \Cref{itm:Kol prop_2} that $\tvdist{E-\Dcal_\Psi}\ge1-\delta/2$, so there exists an event $\Ecal$ with mass at least $1-\delta$ in $E$ but at most $\delta$ in $\Dcal_\Psi$. Moreover $f(\Ucal^m)=aE+(1-a)X$, so this event $\Ecal$ has mass at least $a(1-\delta)$ in $f(\Ucal^m)$. Hence, 
\[
    \tvdist{f(\Ucal^m)-\Dcal_\Psi} \ge a(1-\delta) - \delta \geq a/4,
\]
where we have used the fact that $a\ge4\delta$ and $\delta\le1/2$.
This establishes the upper bound on $a$ and completes our proof.
\end{proof}

\subsection{Continuity Bound}

Here we prove the following proposition. It will later be combined with the Kolmogorov bound to prove our local limit theorem.
\begin{proposition*}[\Cref{continuity prop} Restated]
Let $f\colon\{0,1\}^m\rightarrow \{0,1\}^n$ be a $d$-local function with $n$ sufficiently large (in terms of $d$). Let $\Psi \subseteq \{0,1,\ldots,n\}$ be a set containing an element $n(1/2\pm c(d))$ for some $c(d)>0$ a small enough function of $d$. Then the distribution $f(\Ucal^m)$ can be written as a mixture $aE + (1-a)X$ with $a=O(\tvdist{f(\Ucal^m) - \Dcal_\Psi})$ so that for any even $\Delta$ and $x\in\{0,1,\ldots,n\}$,
$$
\big|\Pr\sbra{|X| = x} - \Pr\sbra{|X| = x+\Delta} \big| = O_d\pbra{\frac{|\Delta|}{n}}.
$$
\end{proposition*}

The proof requires machinery from \cite{kane2024locality}. In particular, we recall the alternative viewpoint of a function as a hypergraph $G$ on the output bits $[n]$ with an edge for each input bit containing all of the output bits that depend on it. Observe the locality assumption implies $G$ has maximum degree at most $d$. As a consequence of \cite[Corollary 4.11]{kane2024locality} we have:

\begin{corollary}\label{neighborhoods cor}
Let $G$ be a hypergraph on $n$ vertices with maximum degree at most $d$. For any increasing function $F\colon\mathbb{N}\rightarrow \mathbb{N}$, there exists a set $S$ of edges in $G$ whose removal yields at least $r = n/O_{d,F}(1)$\footnote{Recall $O_{d,F}(1)$ denotes a quantity whose value is constant once $d$ and $F$ are fixed.} vertices in $G$ whose neighborhoods have size at most $t=O_{d,F}(1)$ and are pairwise non-adjacent, and satisfies $|S| \leq r/F(t)$.
\end{corollary}

We will also need the following technical density comparison result, whose full generality and proof are deferred to \Cref{app:comp_llt}.

\begin{theorem}[{Special case of \Cref{thm:comp_llt_C=1}}]\label{thm:comp_llt_C=1_special}
Let $t\ge1$ be an integer, and let $X_1,\ldots,X_n$ be independent random variables in $\cbra{0,1,\ldots,t}$.
For each $i\in[n]$ and integer $r\ge1$, define $p_{r,i}=\max_{x\in\Zbb}\Pr\sbra{X_i\equiv x\Mod r}$ and assume
\begin{equation*}
\sum_{i\in[n]}(1-p_{r,i})\ge \Omega_t(n)
\quad\text{holds for all $r\in\cbra{3,4,\dots, t}$.}
\end{equation*}
Then for any $x\in\Zbb$ and even $\Delta\in\Zbb$, we have
$$
\Pr\sbra{\sum_{i\in[n]}X_i=x}-\Pr\sbra{\sum_{i\in[n]}X_i=x+\Delta}\le O_t\pbra{\frac{|\Delta|}{n}}.
$$
\end{theorem}

\begin{proof}[Proof of \Cref{continuity prop}]
Let $S$ be the set of input coordinates promised by \Cref{neighborhoods cor}, taking $F(t)$ to be a sufficiently large multiple of $2^{2dt}$. Each setting of the variables in $S$ produces a subcube of the inputs. We call a subcube $C$ \emph{weird} if for at least half of the neighborhoods of outputs promised by \Cref{neighborhoods cor}, the distribution of $f(\Ucal(C))$ on those outputs is not uniform.
\begin{claim}
    Any weird subcube $C$ satisfies
    \[
        \tvdist{f(\Ucal(C)) - \Dcal_\Psi} > 1-2^{-\Omega(r \cdot 2^{-2dt})}.
    \]
\end{claim}
\begin{proof}
By \Cref{neighborhoods cor}, the outputs of $f(\Ucal(C))$ on any two of the resulting neighborhoods are independent. Moreover, each neighborhood only depends on at most $dt$ input bits, so if its corresponding marginal distribution is not uniform, it must be at least $2^{-dt}$-far from uniform. Furthermore, \Cref{fct:individual_binom} and \Cref{fct:entropy} imply
\[
    \eta \coloneqq \min_x \frac{\Ucal^n(x)}{\Dcal_\Psi(x)} \ge \frac{\binom{n}{n\pbra{\frac{1}{2}\pm c(d)}}}{2^n} \ge \frac{2^{n\pbra{\Hcal\pbra{\frac{1}{2}\pm c(d)}-1}}}{\sqrt{8n\pbra{\frac{1}{4} - c(d)^2}}} \ge 2^{-\Omega\pbra{n\cdot c(d)^2}}.
\]
Choose $c(d)$ to be at most a small multiple of $\pbra{\frac{r}{n}}^{1/2} \cdot 2^{-dt}$. The claim then follows from applying \Cref{lem:tvdist_after_product} to deduce
\[
    \tvdist{f(\Ucal(C)) - \Dcal_\Psi} \ge 1 - 2\cdot e^{-2^{-2dt-2} \cdot r}/\eta > 1-2^{-\Omega(r \cdot 2^{-2dt})}. \qedhere
\]
\end{proof}

Let $w$ be the fraction of cubes (resulting from conditioning on the variables in $S$) that are weird. Combining our bound $|S| \le r/F(t)$ and our choice of $F = \Omega(2^{2dt})$ with  \Cref{lem:tvdist_after_conditioning} yields
\[
    \tvdist{f(\Ucal^m) - \Dcal_\Psi} \ge w\pbra{1-2^{|S|-\Omega(r \cdot 2^{-2dt})}} \ge \frac{w}{2}.
\]
Thus in our conclusion, we can take $a=w$, $E$ to be the mixture of $f$ applied to the weird subcubes, and $X$ to be $f$ applied to the mixture of the non-weird subcubes. It remains to show that
$$
\big|\Pr\sbra{|f(\Ucal(C))| = x} - \Pr\sbra{|f(\Ucal(C))| = x+\Delta} \big| = O_d\pbra{\frac{|\Delta|}{n}}
$$
for any non-weird subcube $C$. Taking the mixture will yield our result.

For a non-weird subcube $C$, consider the $r$ non-adjacent neighborhoods promised to us by \Cref{neighborhoods cor}, and remove the (at most $r/2$) ones for which the output distribution is not uniform. By renaming which neighborhoods we are considering and decreasing $r$ by a factor of $2$, we may assume that no such neighborhoods exist.  

Each neighborhood is $N(v_i)$ for some central element $v_i$. Take a random assignment of all of the input bits that do not affect some central element, which we call \emph{extraneous inputs}. This fixes the value of every output bit not in one of these neighborhoods, and the weight of the output bits of $v_i$'s neighborhood becomes a random variable $X_i$. Note that the total weight of the output of $f$ is some constant plus the sum of the $X_i$'s, which are independent. We would like to claim that \Cref{thm:comp_llt_C=1_special} can be applied to this situation with high probability.

From here, we proceed similarly to \cite[Claims 5.16 \& 5.23]{kane2024locality}. In particular, consider $v_i$'s neighborhood $N_i$ for some $i$ and some integer modulus $3\leq s \leq t$.
Observe that because the distribution over $f(\Ucal(C))[N_i]$ is uniform, the distribution of the weight modulo $s$ conditioned on the $v_i^{th}$ coordinate being $1$ and conditioned on the $v_i^{th}$ coordinate being $0$ are not the same. On the other hand, conditioning on the extraneous inputs, this coordinate is equally likely to be $0$ as $1$. 
Hence it cannot be the case that $X_i$ is always constant mod $s$, as this would imply that both the distribution of $|f(\Ucal(C))[N_i]| \pmod{s}$ conditioned on $f(\Ucal(C))[v_i] = 0$ and the distribution conditioned on $f(\Ucal(C))[v_i] = 1$ would be equal to the distribution of $X_i \pmod{s}$. 

However, after fixing the extraneous inputs, $X_i$ only depends on at most $d$ input bits. Thus, if it is not constant mod $s$, it must be at least $2^{-d}$-far from constant. 
Furthermore, the bits in the neighborhood only depend on at most $dt$ input bits, so with probability at least $2^{-dt}$ over the choice of values for the extraneous bits, $X_i$ is at least $2^{-d}$-far from constant mod $s$. 
Finally, note that since the neighborhoods are non-adjacent (after removing the edges in $S$), the extraneous bits used to determine $X_i$ are disjoint from the extraneous bits used to determine $X_j$ for $i\neq j$.
Thus, whether or not $X_i$ is constant mod $s$ is independent of whether $X_j$ is.

By Chernoff's inequality (\Cref{fct:chernoff}) and a union bound, except with probability $2^{-\Omega_d(n)}$ over the values of the extraneous bits, we have that there are $\Omega_d(n)$ neighborhoods where $X_i$ is at least $2^{-d}$-far from constant mod $s$ for each $3\leq s \leq t$. We can now apply \Cref{thm:comp_llt_C=1_special} to show that if this event occurs, then the corresponding subcube $C'$ defined by fixing the bits in $S$ and the extraneous bits satisfies
$$
\left|\Pr\sbra{|f(\Ucal(C'))| = x} - \Pr\sbra{|f(\Ucal(C'))| = x+\Delta} \right| = O_d\pbra{\frac{|\Delta|}{n}}.
$$
Taking the mixture over all such subcubes gives our result.
\end{proof}

\subsection{Putting it Together}

We are now prepared to prove \Cref{LLT Theorem} by combining our previous Kolmogorov and continuity bounds. Recall we wish to show the existence of a distribution $\Mcal$ which is a mixture of $\Deven$ and $\Dodd$ such that $\tvdist{f(\Ucal^m) - \Mcal} \le O(\tvdist{f(\Ucal^m) - \Dcal_\Psi}) + \delta$ for any $\delta > 0$, provided $n$ is sufficiently large in terms of $d$ and $\delta$. 

\begin{proof}[Proof of \Cref{LLT Theorem}]
Let $c > 0$ be a small constant, and set $\delta' = \max\cbra{\tvdist{f(\Ucal^m) - \Dcal_\Psi}/c, \delta}.$ Then we have
\begin{equation}\label{eq:assume_at_most_Odelta}
    \tvdist{f(\Ucal^m) - \Dcal_\Psi} \le c\delta'.
\end{equation}
Note for any distribution $\Mcal$ which is a mixture of $\Deven$ and $\Dodd$, two applications of \Cref{lem:distance_to_sym} yield
\begin{align*}
\tvdist{f(\Ucal^m)-\Mcal} & = \Theta(\tvdist{|f(\Ucal^m)|-|\Mcal|}) + \Theta(\tvdist{f(\Ucal^m)_\sym-f(\Ucal^m)})\\
& = \Theta(\tvdist{|f(\Ucal^m)| - |\Mcal|}) + \Theta(\tvdist{f(\Ucal^m) - \Dcal_\Psi} - \tvdist{|f(\Ucal^m)| - |\Dcal_\Psi|})\\
& = \Theta(\tvdist{|f(\Ucal^m)| - |\Mcal|}) + O(\tvdist{f(\Ucal^m) - \Dcal_\Psi})
\end{align*}
Thus, it suffices to construct a suitable $\Mcal$ and show $\tvdist{|f(\Ucal^m)| - |\Mcal|} = O(\delta'')$ for some $\delta'' = \Theta(\delta')$ with a sufficiently small implicit constant, since
\[
    \Theta\pbra{\tvdist{|f(\Ucal^m)| - |\Mcal|}} = O(\delta'') \le \delta' \le \frac{1}{c}\cdot \tvdist{f(\Ucal^m) - \Dcal_\Psi} + \delta.
\]

We begin by applying \Cref{Kol prop} with the $\delta$ from that proposition taken to be $\kappa^2$ for some $\kappa>0$  sufficiently small in terms of $d$ and $\delta''$. Let $\Mcal = \eta\cdot\Deven + (1-\eta)\cdot\Dodd$ for the $\eta$ coming from that result. Next, we partition $\{0,1,\ldots,n\}$ into \emph{parity intervals}, each of which is an interval of length $\Theta(\kappa \sqrt{n})$ intersected with either the set of even integers or the set of odd integers. In particular, we have by \Cref{Kol prop} and \Cref{eq:assume_at_most_Odelta} that $f(\Ucal^m)$ is $O(\delta'')$-close to a distribution $X$ so that for any parity interval $\Ical$, we have that
\begin{equation}\label{eq:XvsMinI}
    \big|\Pr\sbra{|X| \in \Ical} - \Pr\sbra{|\Mcal| \in \Ical}\big| = O(\kappa^2).
\end{equation}

We next apply \Cref{continuity prop} to find that $f(\Ucal^m)$ is $O(\delta'')$-close to a distribution $X'$ satisfying
\begin{equation}\label{eq:combining_nearby_weights}
    \big|\Pr\sbra{|X'| = x} - \Pr\sbra{|X'|=x+\Delta}\big| = O_d(\Delta/n)
\end{equation}
for any $x$ and even $\Delta$. Note also that $\tvdist{X-X'} = O(\delta'')$. We claim that this is enough to show that $\tvdist{|X'|-|\Mcal|}$ is small. In particular, we have that 
\begin{align}\label{sum over intervals equation}
\notag \tvdist{|X'|-|\Mcal|} & = \frac{1}{2} \sum_{x=0}^n |\Pr\sbra{|X'|=x} - \Pr\sbra{|\Mcal|=x}|\\
\notag & = \frac{1}{2}\sum_{\Ical} \sum_{x\in \Ical}|\Pr\sbra{|X'|=x} - \Pr\sbra{|\Mcal|=x}|\\
\notag & \leq \sum_{\Ical} \sum_{x\in \Ical}\left|\Pr\sbra{|X'|=x} - \frac{\Pr\sbra{|X'| \in \Ical}}{|\Ical|}\right| + \left|\frac{\Pr\sbra{|X'| \in \Ical}-\Pr\sbra{|\Mcal| \in \Ical}}{|\Ical|}\right| \\
& \ \ \ \ \ \ \ \ \ \ \ \ \ \ \ \  + \left|\Pr\sbra{|\Mcal|=x} - \frac{\Pr\sbra{|\Mcal| \in \Ical}}{|\Ical|}\right|.
\end{align}

To analyze this, we note that \Cref{fct:hoeffding} implies all but $O(\delta'')$ of the mass of $|\Mcal|$ is supported on $O(\log(1/\delta'')/\kappa)$ many parity intervals, which we call \emph{big}. The total mass that $|X|$ assigns to these big intervals is
\begin{align*}
\sum_{\textrm{big }\Ical} \Pr\sbra{|X| \in \Ical} & = \sum_{\textrm{big }\Ical} \pbra{\Pr\sbra{|\Mcal| \in \Ical} + O(\kappa^2)} \tag{by \Cref{eq:XvsMinI}} \\
& = 1- O(\delta'') + O(\kappa^2 \log(1/\delta'')/\kappa) = 1-O(\delta''),
\end{align*}
with the final equality holding for $\kappa$ sufficiently small.
Therefore, since $\tvdist{X-X'}= O(\delta'')$, $X'$ also assigns all but an $O(\delta'')$-fraction of its mass to big intervals. Hence, up to this $O(\delta'')$ error, we can restrict the sum in \Cref{sum over intervals equation} to big intervals. Thus, 
\begin{align*}
\tvdist{|X'|-|\Mcal|} & \leq O(\delta'')+ \sum_{\textrm{big }\Ical} \sum_{x\in \Ical}\left|\Pr\sbra{|X'|=x} - \frac{\Pr\sbra{|X'| \in \Ical}}{|\Ical|}\right| + \left|\frac{\Pr\sbra{|X'| \in \Ical}-\Pr\sbra{|\Mcal| \in \Ical}}{|\Ical|}\right| \\
& \ \ \ \ \ \ \ \ \ \ \ \ \ \ \ \  + \left|\Pr\sbra{|\Mcal|=x} - \frac{\Pr\sbra{|\Mcal| \in \Ical}}{|\Ical|}\right|.
\end{align*}
Clearly,
\begin{align*}
    \sum_{x\in \Ical}\left|\frac{\Pr\sbra{|X'| \in \Ical}-\Pr\sbra{|\Mcal| \in \Ical}}{|\Ical|}\right| &= |\Pr\sbra{|X'| \in \Ical}-\Pr\sbra{|\Mcal| \in \Ical}| \\
    &\le |\Pr\sbra{|X'| \in \Ical}-\Pr\sbra{|X| \in \Ical}|+ |\Pr\sbra{|X| \in \Ical}-\Pr\sbra{|\Mcal| \in \Ical}|.
\end{align*}
The first term summed over all $\Ical$ is at most $\tvdist{X-X'}=O(\delta'')$. The second term is at most $O(\kappa^2)$ by \Cref{eq:XvsMinI}, so summed over all big intervals is $O(\delta'')$ (for small enough $\kappa$).

Additionally, note that
$$
\left|\Pr\sbra{|X'|=x} - \frac{\Pr\sbra{|X'| \in \Ical}}{|\Ical|}\right| \le
\max_{y\in \Ical} |\Pr\sbra{|X'|=x} - \Pr\sbra{|X'|=y}|.
$$
Since $x$ and $y$ have the same parity in any parity interval, \Cref{eq:combining_nearby_weights} implies this is at most
$$
\max_{y\in \Ical} O_d\pbra{\frac{|x-y|}{n}} = O_d\pbra{\frac{\kappa \sqrt{n}}{n}} = O_d\pbra{\frac{\kappa}{\sqrt{n}}}.
$$
Summing this over all $x\in \Ical$ gives $O_d(\kappa^2)$, and summing over all big intervals gives $O_d(\kappa \log(1/\delta'')) = O(\delta'')$. The sum of the 
$$
\left|\Pr\sbra{|\Mcal|=x} - \frac{\Pr\sbra{|\Mcal| \in \Ical}}{|\Ical|}\right|
$$
terms can be bounded similarly.
We infer that $\tvdist{|X'|-|\Mcal|} = O(\delta'')$, and thus by the triangle inequality
\[
    \tvdist{|f(\Ucal^m)|-|\Mcal|} \le \tvdist{|f(\Ucal^m)|-|X|} + \tvdist{|X|-|X'|} + \tvdist{|X'|-|\Mcal|} = O(\delta''),
\]
completing our proof.
\end{proof}

\section*{Acknowledgments}
We thank anonymous reviewers for helpful comments on an earlier version of this manuscript, and we thank Jason Gaitonde for pointing us to \cite{chattopadhyay2020xor} on MathOverflow \cite{mathoverflow}. KW wants to thank Tiancheng He for answering questions in probability theory. AO wants to thank Farzan Byramji, Shachar Lovett, and Jackson Morris for suggesting references about low-depth circuits.

\bibliographystyle{alpha} 
\bibliography{biblio}

\newcommand{\etalchar}[1]{$^{#1}$}
\begin{thebibliography}{GGH{\etalchar{+}}07}

\bibitem[Bab87]{babai1987random}
L{\'a}szi{\'o} Babai.
\newblock Random oracles separate {PSPACE} from the polynomial-time hierarchy.
\newblock {\em Information Processing Letters}, 26(1):51--53, 1987.

\bibitem[BDK{\etalchar{+}}13]{beyersdorff2013verifying}
Olaf Beyersdorff, Samir Datta, Andreas Krebs, Meena Mahajan, Gido Scharfenberger-Fabian, Karteek Sreenivasaiah, Michael Thomas, and Heribert Vollmer.
\newblock Verifying proofs in constant depth.
\newblock {\em ACM Transactions on Computation Theory (TOCT)}, 5(1):1--23, 2013.

\bibitem[BGK18]{bravyi2018quantum}
Sergey Bravyi, David Gosset, and Robert K{\"o}nig.
\newblock Quantum advantage with shallow circuits.
\newblock {\em Science}, 362(6412):308--311, 2018.

\bibitem[BGKT20]{bravyi2020quantum}
Sergey Bravyi, David Gosset, Robert K{\"o}nig, and Marco Tomamichel.
\newblock Quantum advantage with noisy shallow circuits.
\newblock {\em Nature Physics}, 16(10):1040--1045, 2020.

\bibitem[BIL12]{beck2012large}
Chris Beck, Russell Impagliazzo, and Shachar Lovett.
\newblock Large deviation bounds for decision trees and sampling lower bounds for {AC0}-circuits.
\newblock In {\em 2012 IEEE 53rd Annual Symposium on Foundations of Computer Science}, pages 101--110. IEEE, 2012.

\bibitem[BL87]{boppana1987one}
Ravi~B Boppana and Jeffrey~C Lagarias.
\newblock One-way functions and circuit complexity.
\newblock {\em Information and Computation}, 74(3):226--240, 1987.

\bibitem[Bon70]{bonami1970etude}
Aline Bonami.
\newblock {\'E}tude des coefficients de fourier des fonctions de {$L^{p}(G)$}.
\newblock In {\em Annales de l'institut Fourier}, volume~20, pages 335--402, 1970.

\bibitem[Bru12]{bezout_mo}
Fran\c{o}is Brunault.
\newblock Estimates for {B}ezout coefficients.
\newblock MathOverflow, 2012.
\newblock URL:https://mathoverflow.net/q/108723 (version: 2012-10-05).

\bibitem[CGZ22]{chattopadhyay2022space}
Eshan Chattopadhyay, Jesse Goodman, and David Zuckerman.
\newblock The space complexity of sampling.
\newblock In {\em 13th Innovations in Theoretical Computer Science Conference,(ITCS 2022)}, 2022.

\bibitem[CHH{\etalchar{+}}20]{chattopadhyay2020xor}
Eshan Chattopadhyay, Pooya Hatami, Kaave Hosseini, Shachar Lovett, and David Zuckerman.
\newblock {XOR} lemmas for resilient functions against polynomials.
\newblock In {\em Proceedings of the 52nd Annual ACM SIGACT Symposium on Theory of Computing}, pages 234--246, 2020.

\bibitem[CS16]{cohen2016extractors}
Gil Cohen and Leonard~J Schulman.
\newblock Extractors for near logarithmic min-entropy.
\newblock In {\em 2016 IEEE 57th Annual Symposium on Foundations of Computer Science (FOCS)}, pages 178--187. IEEE, 2016.

\bibitem[CT06]{cover2006elements}
Thomas~M Cover and Joy~A Thomas.
\newblock Elements of information theory, 2006.

\bibitem[CZ16]{chattopadhyay2016explicit}
Eshan Chattopadhyay and David Zuckerman.
\newblock Explicit two-source extractors and resilient functions.
\newblock In {\em Proceedings of the forty-eighth annual ACM symposium on Theory of Computing}, pages 670--683, 2016.

\bibitem[DGJ{\etalchar{+}}10]{diakonikolas2010bounded}
Ilias Diakonikolas, Parikshit Gopalan, Ragesh Jaiswal, Rocco~A Servedio, and Emanuele Viola.
\newblock Bounded independence fools halfspaces.
\newblock {\em SIAM Journal on Computing}, 39(8):3441--3462, 2010.

\bibitem[DW12]{de2012extractors}
Anindya De and Thomas Watson.
\newblock Extractors and lower bounds for locally samplable sources.
\newblock {\em ACM Transactions on Computation Theory (TOCT)}, 4(1):1--21, 2012.

\bibitem[FLRS23]{filmus2023sampling}
Yuval Filmus, Itai Leigh, Artur Riazanov, and Dmitry Sokolov.
\newblock Sampling and certifying symmetric functions.
\newblock In {\em Approximation, Randomization, and Combinatorial Optimization. (APPROX/RANDOM)}, 2023.

\bibitem[Gai23]{mathoverflow}
Jason Gaitonde.
\newblock Are there a few input bits that randomize the output of an $\mathbb{F}_2$ polynomial?
\newblock MathOverflow, 2023.
\newblock URL:https://mathoverflow.net/q/460879 (version: 2023-12-22).

\bibitem[GGH{\etalchar{+}}07]{goldwasser2007verifying}
Shafi Goldwasser, Dan Gutfreund, Alexander Healy, Tali Kaufman, and Guy~N Rothblum.
\newblock Verifying and decoding in constant depth.
\newblock In {\em Proceedings of the thirty-ninth annual ACM symposium on Theory of computing}, pages 440--449, 2007.

\bibitem[GGNS23]{gajulapalli2023range}
Karthik Gajulapalli, Alexander Golovnev, Satyajeet Nagargoje, and Sidhant Saraogi.
\newblock Range avoidance for constant-depth circuits: Hardness and algorithms.
\newblock {\em arXiv preprint arXiv:2303.05044}, 2023.

\bibitem[GLW22]{guruswami2022range}
Venkatesan Guruswami, Xin Lyu, and Xiuhan Wang.
\newblock Range avoidance for low-depth circuits and connections to pseudorandomness.
\newblock In {\em Approximation, Randomization, and Combinatorial Optimization. Algorithms and Techniques (APPROX/RANDOM 2022)}. Schloss-Dagstuhl-Leibniz Zentrum f{\"u}r Informatik, 2022.

\bibitem[GW20]{goos2020lower}
Mika G{\"o}{\"o}s and Thomas Watson.
\newblock A lower bound for sampling disjoint sets.
\newblock {\em ACM Transactions on Computation Theory (TOCT)}, 12(3):1--13, 2020.

\bibitem[Hag91]{hagerup1991fast}
Torben Hagerup.
\newblock Fast parallel generation of random permutations.
\newblock In {\em Automata, Languages and Programming: 18th International Colloquium Madrid, Spain, July 8--12, 1991 Proceedings 18}, pages 405--416. Springer, 1991.

\bibitem[H{\aa}s86]{haastad1986computational}
Johan H{\aa}stad.
\newblock {\em Computational limitations for small depth circuits}.
\newblock PhD thesis, Massachusetts Institute of Technology, 1986.

\bibitem[JVV86]{jerrum1986random}
Mark~R Jerrum, Leslie~G Valiant, and Vijay~V Vazirani.
\newblock Random generation of combinatorial structures from a uniform distribution.
\newblock {\em Theoretical computer science}, 43:169--188, 1986.

\bibitem[Kan17]{kane2017structure}
Daniel~M Kane.
\newblock A structure theorem for poorly anticoncentrated polynomials of {G}aussians and applications to the study of polynomial threshold functions.
\newblock 2017.

\bibitem[KKL17]{kabanets2017polynomial}
Valentine Kabanets, Daniel~M Kane, and Zhenjian Lu.
\newblock A polynomial restriction lemma with applications.
\newblock In {\em Proceedings of the 49th Annual ACM SIGACT Symposium on Theory of Computing}, pages 615--628, 2017.
\newblock Available at \url{https://cseweb.ucsd.edu/~dakane/PTFblockRestriction.pdf}.

\bibitem[KLMS16]{krebs2016small}
Andreas Krebs, Nutan Limaye, Meena Mahajan, and Karteek Sreenivasaiah.
\newblock Small depth proof systems.
\newblock {\em ACM Transactions on Computation Theory (TOCT)}, 9(1):1--26, 2016.

\bibitem[KOW24]{kane2024locality}
Daniel~M Kane, Anthony Ostuni, and Kewen Wu.
\newblock Locality bounds for sampling {H}amming slices.
\newblock In {\em Proceedings of the 56th Annual ACM Symposium on Theory of Computing}, pages 1279--1286, 2024.
\newblock Available at \url{https://arxiv.org/abs/2402.14278}.

\bibitem[Lug17]{binomial_sum_mo}
Michael Lugo.
\newblock Sum of the first k binomial coefficients for fixed $n$.
\newblock MathOverflow, 2017.
\newblock URL:https://mathoverflow.net/q/17236 (version: 2017-10-01).

\bibitem[LV11]{lovett2011bounded}
Shachar Lovett and Emanuele Viola.
\newblock Bounded-depth circuits cannot sample good codes.
\newblock In {\em 2011 IEEE 26th Annual Conference on Computational Complexity}, pages 243--251. IEEE, 2011.

\bibitem[MV91]{matias1991converting}
Yossi Matias and Uzi Vishkin.
\newblock Converting high probability into nearly-constant time—with applications to parallel hashing.
\newblock In {\em Proceedings of the twenty-third annual ACM symposium on Theory of Computing}, pages 307--316, 1991.

\bibitem[RSW22]{ren2022range}
Hanlin Ren, Rahul Santhanam, and Zhikun Wang.
\newblock On the range avoidance problem for circuits.
\newblock In {\em 2022 IEEE 63rd Annual Symposium on Foundations of Computer Science (FOCS)}, pages 640--650. IEEE, 2022.

\bibitem[SS24]{shaltiel2024explicit}
Ronen Shaltiel and Jad Silbak.
\newblock Explicit codes for poly-size circuits and functions that are hard to sample on low entropy distributions.
\newblock In {\em Proceedings of the 56th Annual ACM Symposium on Theory of Computing}, pages 2028--2038, 2024.

\bibitem[Ush86]{ushakov1986upper}
Nikolai~G Ushakov.
\newblock Upper estimates of maximum probability for sums of independent random vectors.
\newblock {\em Theory of Probability \& Its Applications}, 30(1):38--49, 1986.

\bibitem[Vio12a]{viola2012bit}
Emanuele Viola.
\newblock Bit-probe lower bounds for succinct data structures.
\newblock {\em SIAM Journal on Computing}, 41(6):1593, 2012.

\bibitem[Vio12b]{viola2012complexity}
Emanuele Viola.
\newblock The complexity of distributions.
\newblock {\em SIAM Journal on Computing}, 41(1):191--218, 2012.

\bibitem[Vio12c]{viola2012extractors}
Emanuele Viola.
\newblock Extractors for {T}uring-machine sources.
\newblock In {\em International Workshop on Approximation Algorithms for Combinatorial Optimization}, pages 663--671. Springer, 2012.

\bibitem[Vio14]{viola2014extractors}
Emanuele Viola.
\newblock Extractors for circuit sources.
\newblock {\em SIAM Journal on Computing}, 43(2):655--672, 2014.

\bibitem[Vio16]{viola2016quadratic}
Emanuele Viola.
\newblock Quadratic maps are hard to sample.
\newblock {\em ACM Transactions on Computation Theory (TOCT)}, 8(4):1--4, 2016.

\bibitem[Vio20]{viola2020sampling}
Emanuele Viola.
\newblock Sampling lower bounds: boolean average-case and permutations.
\newblock {\em SIAM Journal on Computing}, 49(1):119--137, 2020.

\bibitem[Vio23]{viola2023new}
Emanuele Viola.
\newblock New sampling lower bounds via the separator.
\newblock In {\em 38th Computational Complexity Conference (CCC 2023)}. Schloss Dagstuhl-Leibniz-Zentrum f{\"u}r Informatik. Available at \url{https://eccc.weizmann.ac.il/report/2021/073/}, 2023.

\bibitem[Wik23a]{wiki:Binary_entropy_function}
Wikipedia.
\newblock {Binary entropy function} --- {W}ikipedia{,} the free encyclopedia.
\newblock \url{http://en.wikipedia.org/w/index.php?title=Binary\%20entropy\%20function&oldid=1071507954}, 2023.
\newblock [Online; accessed 04-December-2023].

\bibitem[Wik23b]{wiki:Binomial_coefficient}
Wikipedia.
\newblock {Binomial coefficient} --- {W}ikipedia{,} the free encyclopedia.
\newblock \url{http://en.wikipedia.org/w/index.php?title=Binomial\%20coefficient&oldid=1187835533}, 2023.
\newblock [Online; accessed 15-December-2023].

\bibitem[Wik23c]{wiki:Bézout's_identity}
Wikipedia.
\newblock {Bézout's identity} --- {W}ikipedia{,} the free encyclopedia.
\newblock \url{http://en.wikipedia.org/w/index.php?title=B\%C3\%A9zout's\%20identity&oldid=1179305736}, 2023.
\newblock [Online; accessed 05-December-2023].

\bibitem[WKST19]{watts2019exponential}
Adam~Bene Watts, Robin Kothari, Luke Schaeffer, and Avishay Tal.
\newblock Exponential separation between shallow quantum circuits and unbounded fan-in shallow classical circuits.
\newblock In {\em Proceedings of the 51st Annual ACM SIGACT Symposium on Theory of Computing}, pages 515--526, 2019.

\bibitem[WP23]{watts2023unconditional}
Adam~Bene Watts and Natalie Parham.
\newblock Unconditional quantum advantage for sampling with shallow circuits.
\newblock {\em arXiv preprint arXiv:2301.00995}, 2023.

\end{thebibliography}

\appendix
\section{Density Comparison of Sum of Integral Random Variables}\label{app:comp_llt}

The goal of this section is to prove the following density comparison result for sums of integral random variables.

\begin{theorem}\label{thm:comp_llt_C=1}
Let $t\ge1$ be an integer, and let $X_1,\ldots,X_n$ be independent random variables in $\cbra{0,1,\ldots,t}$.
Let $\Phi\subseteq\cbra{2,3,\ldots,t}$.
Define $\phi$ as the least common multiple of values in $[t]\setminus\Phi$.

For each $i\in[n]$ and integer $r\ge1$, define $p_{r,i}=\max_{x\in\Zbb}\Pr\sbra{X_i\equiv x\Mod r}$ and assume\footnote{Note that if \Cref{eq:thm:comp_llt_C=1_1} holds for some $r$, then it also holds for $r'$ that is a multiple of $r$ as $\Pr\sbra{X_i\equiv x\Mod{r'}}\le\Pr\sbra{X_i\equiv x\Mod r}$. Hence we may assume that $\Phi$ contains all the multiples of $r$ (up to $t$) if $r\in\Phi$.}
\begin{equation}\label{eq:thm:comp_llt_C=1_1}
\sum_{i\in[n]}(1-p_{r,i})\ge L>0
\quad\text{holds for all $r\in\Phi$.}
\end{equation}
Let $\alpha=\pbra{\frac L{4n(t+1)}}^{t^2\phi}$ and assume $m\coloneqq \floorbra{L/(16t^4\phi)} \ge 1$.
Then for any $x\in\Zbb$, we have
$$
\Pr\sbra{\sum_{i\in[n]}X_i=x}-\Pr\sbra{\sum_{i\in[n]}X_i=x+\Delta}\le\frac{22|\Delta|}{\phi\cdot\alpha m}
$$
holds for any $\Delta\in\Zbb$ that is a multiple of $\phi$.
\end{theorem}

The typical setting for \Cref{thm:comp_llt_C=1} is when we have small $t$ and $L=\Theta_t(n)$; then $\alpha$ is also a constant depending only on $t$.

\begin{remark}
The assumption of $\Delta$ being a multiple of $\phi$ is necessary.
If the $X_i$'s have some joint congruence relation not shared with $\Delta$, the bound can fail.
Consider the case where $n$ is even and each $X_i$ is uniform in $\cbra{1,3}$, which violates \Cref{eq:thm:comp_llt_C=1_1} only for $r=2$.
Then we set $x=2n$ and $\Delta=1$.
Since the sum is $n$ plus twice an $n$-bit binomial distribution, we have $\Pr\sbra{\sum_iX_i=x}\approx1/\sqrt n$ but $\Pr\sbra{\sum_iX_i=x+\Delta}=0$.
However, $m \approx n$ and $\alpha$ is a constant.
Hence the final bound does not hold.

We also note that the quantitative bound of $\alpha$ and $m$ can be slightly improved by tightening our analysis.
Since it does not change our final bounds by much, we choose the cleaner presentation here.
\end{remark}

We will need the following simple bound on the difference of nearby binomial coefficients.

\begin{fact}\label{fct:nearby_binom}
For any integers $n,b\ge1$, we have $2^{-n}\cdot\pbra{\binom nb-\binom n{b+1}}\le\frac7n$.
Moreover for any integer $\Delta\ge0$, we have
$$
2^{-n}\cdot\pbra{\binom nb-\binom n{b+\Delta}}\le\frac{7\Delta}n.
$$
\end{fact}
\begin{proof}
The moreover part follows from a telescoping sum. Hence we focus on the first bound and divide into the following cases:
\begin{itemize}
\item If $b<n/2$, then $\binom nb\le\binom n{b+1}$ and the bound trivially holds.
\item If $b\ge n$, then $\binom nb\le1$, $\binom n{b+1}=0$, and the bound trivially holds.
\item If $n/2\le b\le n-1$, then 
\begin{align*}
\binom nb-\binom n{b+1}
&=\binom nb\cdot\pbra{1-\frac{n-b}{b+1}}=\binom nb\cdot\frac{2b+1-n}{b+1}\\
&\le\frac{2^{n\cdot\Hcal(b/n)}}{\sqrt{\pi b(1-b/n)}}\cdot\frac{2b+1-n}{b+1}
\tag{by \Cref{fct:individual_binom}}\\
&\le\frac{2^{n\cdot\Hcal(b/n)}}{\sqrt{n/2}}\cdot\frac{2b+1-n}{n/2}.
\tag{since $n/2\le b\le n-1$}
\end{align*}
Define $x=\frac{2b}n-1$. Then $x\in[0,1)$.
By \Cref{fct:entropy}, we have $\Hcal(b/n)-1=\Hcal\pbra{\frac{1+x}2}-1\le-\frac{x^2}{2\ln(2)}$ and hence
\begin{align*}
2^{-n}\cdot\pbra{\binom nb-\binom n{b+1}}
&\le2^{-n\cdot\frac{x^2}{2\ln(2)}}\cdot\frac1{\sqrt{n/2}}\cdot\frac{n\cdot x+1}{n/2}\\
&\le\underbrace{x\cdot2^{-n\cdot\frac{x^2}{2\ln(2)}}}_A\cdot\frac{2\sqrt2}{\sqrt{n}}+\pbra{\frac2n}^{1.5}.
\end{align*}
Writing $x=\sqrt{\frac{2\ln(2)\cdot\log(y)}n}$ for some $y\ge1$, we transform $A$ above into
$$
A=\sqrt{\frac{2\ln(2)\cdot\log(y)}n}\cdot\frac1y\le\sqrt{\frac{2\ln(2)}n}.
$$
Hence $2^{-n}\cdot\pbra{\binom nb-\binom n{b+1}}\le\frac{4\sqrt{\ln(2)}}n+\pbra{\frac2n}^{1.5}\le\frac7n$ as claimed.
\qedhere
\end{itemize}
\end{proof}

To prove \Cref{thm:comp_llt_C=1}, we observe that intuitively $\sum_{i\in[n]}X_i$ should converge to a (discrete) Gaussian distribution with large variance.
Then in this (discrete) Gaussian distribution, 
\begin{itemize}
\item if $x$ lies much outside the standard deviation regime around the mean, then it has small density already;
\item otherwise its density, compared with the density of $x+\Delta$, is only off by a small multiplicative factor, which means the quantity of interest is in fact small.
\end{itemize}

We first prove a simpler case where each random variable always has a ``neighboring'' pair of values in its support.
Note that in this case we do not need to assume that the random variables are bounded.
Later we will reduce the case of \Cref{thm:comp_llt_C=1} to this setting.

\begin{lemma}\label{lem:comp_llt_simp_C=1}
Let $Y_1,\ldots,Y_m$ be independent integer random variables, and let $\phi\ge1$ be an integer.
Assume that $\alpha>0$ is a parameter such that for each $i\in[m]$, there exists $u_i\in\Zbb$ satisfying
$$
\Pr[Y_i=u_i]\ge\alpha
\quad\text{and}\quad
\Pr[Y_i=u_i+\phi]\ge\alpha.
$$
Then for any $y\in\Zbb$, we have
\begin{equation}\label{eq:lem:comp_llt_simp_C=1_2}
\Pr\sbra{\sum_{i\in[m]}Y_i=y}-\Pr\sbra{\sum_{i\in[m]}Y_i=y+\Delta}
\le\frac{22|\Delta|}{\phi\cdot\alpha m}
\end{equation}
holds for any $\Delta\in\Zbb$ that is a multiple of $\phi$.
\end{lemma}

\begin{proof}
The bound trivially holds if $\Delta=0$.
If $\Delta<0$, then we work with negated $Y_i$'s.
Hence we assume $\Delta\ge\phi$.
By subtracting $u_i$ from $Y_i$ and $y$, we assume that each $u_i$ equals zero.
Then we decompose each $Y_i=W_i\cdot B_i+(1-W_i)\cdot Z_i$, where $B_i$ is uniform over $\cbra{0,\phi}$, $W_i$ be an $\alpha$-biased coin (i.e., $\Pr[W_i=1]=\alpha$ and $\Pr[W_i=0]=1-\alpha$), and $Z_i$ is some integer random variable.
In addition, $W_i,B_i,Z_i$ are independent.

Now define $\Ecal$ to be the event that $\sum_{i\in[m]}W_i\le\alpha\cdot m/2$.
Then by \Cref{fct:chernoff} with $\delta=1/2$ and $\mu=\alpha\cdot m$, we have
\begin{equation}\label{eq:lem:comp_llt_simp_1}
\Pr[\Ecal]\le e^{-\alpha\cdot m/8}.
\end{equation}
For fixed $W=(W_1,\ldots,W_m)$ under which $\Ecal$ does not happen, let $S=\cbra{i\in[m] : W_i=1}$ of size $k=|S|\ge\alpha\cdot m/2$.
Then for any fixed $Z=(Z_1,\ldots,Z_m)$, the LHS of \Cref{eq:lem:comp_llt_simp_C=1_2} equals
$$
\Pr\sbra{\sum_{i\in S}B_i=b\mid W,Z,\neg\Ecal}-\Pr\sbra{\sum_{i\in S}B_i=b+\Delta\mid W,Z,\neg\Ecal},
$$
where $b=y-\sum_{i\notin S}Z_i$.
Recall that each $B_i$ is uniform over $\cbra{0,\phi}$.
If $b$ is not a multiple of $\phi$, then the above quantity equals zero since $\Delta$ is a multiple of $\phi$.
Otherwise, let $b'=b/\phi$ and $\Delta'=\Delta/\phi\ge1$.
Then the above quantity equals 
\begin{equation}\label{eq:lem:comp_llt_simp_4}
2^{-k}\cdot\pbra{\binom k{b'}-\binom k{b'+\Delta'}}\le\frac{7\Delta'}k
\end{equation}
by \Cref{fct:nearby_binom}.
This, combined with \Cref{eq:lem:comp_llt_simp_1}, establishes \Cref{eq:lem:comp_llt_simp_C=1_2}: 
\begin{align*}
\text{LHS of \Cref{eq:lem:comp_llt_simp_C=1_2}}
&\le\Pr[\Ecal]+\E\sbra{\indicator_{\phi\text{ divides }b}\cdot2^{-k}\cdot\pbra{\binom k{b'}-\binom k{b'+\Delta'}}\mid\neg\Ecal}\\
&\le e^{-\alpha\cdot m/8}+\E\sbra{\frac{7\Delta'}k\mid\neg\Ecal}
\tag{by \Cref{eq:lem:comp_llt_simp_1} and \Cref{eq:lem:comp_llt_simp_4}}\\
&\le\frac8{\alpha\cdot m}+\frac{14\Delta'}{\alpha\cdot m}
\tag{since $e^{-x}\le\frac1x$ and $k\ge\alpha m/2$}\\
&\le\frac{22\Delta}{\phi\cdot\alpha m}
\tag{since $\Delta'=\Delta/\phi\ge1$}\\
&=\text{RHS of \Cref{eq:lem:comp_llt_simp_C=1_2}}.
\tag*{\qedhere}
\end{align*}
\end{proof}

Now we prove \Cref{thm:comp_llt_C=1} by reducing it to \Cref{lem:comp_llt_simp_C=1}. To this end, we will divide $X_1,\ldots,X_n$ into many parts, and the sum within each part will have two neighboring values with noticeable probability weights.

\begin{proof}[Proof of \Cref{thm:comp_llt_C=1}]
Denote $\Phi=\cbra{r_1,r_2,\ldots,r_k}$ where $k=|\Phi|\le t$.
For each $r_j$, let $S_{r_j}\subseteq[n]$ be the set of $X_i$'s with $1-p_{r_j,i}\ge L/(2n)$, i.e.,
$$
S_{r_j}=\cbra{i\in[n] : 1-p_{r_j,i}\ge L/(2n)}.
$$
Since $0\le1-p_{r_j,i}\le1$ and by \Cref{eq:thm:comp_llt_C=1_1}, we have
$$
L\le\sum_{i\in[n]}(1-p_{r_j,i})\le\abs{S_{r_j}}\cdot1+\pbra{n-\abs{S_{r_j}}}\cdot\frac L{2n},
$$
which implies $\abs{S_{r_j}}\ge L/2$.
Now we remove multiple appearances of indices across $S_{r_j}$'s to make them pairwise disjoint.
Formally, for each $j=1,2,\ldots,k$, we keep $n'\coloneqq \floorbra{L/(4t)}$ elements in $S_{r_j}$ and update $S_{r_{j'}}\gets S_{r_{j'}}\setminus S_{r_j}$ for all $j'>j$.
Since $k\le t$ and originally $\abs{S_{r_j}}\ge L/2$, each $S_{r_j}$ contains at least $n'$ elements after this pruning.

For each $j\in[k]$ and $i\in S_{r_j}$, by an averaging argument, there exists $c_i\in\Zbb/r_j\Zbb$ such that $\Pr\sbra{X_i\equiv c_i\Mod{r_j}}\ge\frac1{r_j}$. Hence, by another averaging argument, there exists $z_i\in\cbra{0,1,\ldots,t}$ such that $z_i\equiv c_i\Mod{r_j}$ and
$$
\Pr\sbra{X_i=z_i}\ge\frac1{r_j}\cdot\frac1{\ceilbra{(t+1)/r_j}}\ge\frac1{2(t+1)}\ge\frac L{4n(t+1)},
$$
where we used the fact that $0<L\le n$.
Since $i\in S_{r_j}$, we also have $\Pr\sbra{X_i\equiv c_i\Mod{r_j}}\le1-\frac L{2n}$ and hence, by an averaging argument, there exists $c_i'\in\Zbb/r_j\Zbb$ such that $c_i'\neq c_i$ and 
$\Pr\sbra{X_i\equiv c_i'\Mod{r_j}}\ge\frac L{2n\cdot(r_j-1)}$.
Similarly by another averaging argument, there exists $z_i'\in\cbra{0,1,\ldots,t}$ such that $z_i'\equiv c_i'\Mod{r_j}$ and 
$$
\Pr\sbra{X_i=z_i'}\ge\frac L{2n\cdot(r_j-1)}\cdot\frac1{\ceilbra{(t+1)/r_j}}\ge\frac L{4n(t+1)}.
$$
Since both $z_i$ and $z_i'$ are in $\cbra{0,1,\ldots,t}$, by a final averaging argument, there exists $z_{r_j},z_{r_j}'$ such that
\begin{enumerate}
\item\label{itm:thm:comp_llt_C=1_1} $z_{r_j},z_{r_j}'\in\cbra{0,1,\ldots,t}$ and $z_{r_j}\not\equiv z_{r_j}'\Mod{r_j}$, and
\item\label{itm:thm:comp_llt_C=1_2} at least a $1/\binom{t+1}2\ge1/t^2$ fraction of $i\in S_{r_j}$ satisfy $\Pr\sbra{X_i=z_{r_j}},\Pr\sbra{X_i=z_{r_j}'}\ge\frac L{4n(t+1)}$.
\end{enumerate}
Let $n''=\ceilbra{n'/t^2}=\ceilbra{|S_{r_j}|/t^2}$.
Based on \Cref{itm:thm:comp_llt_C=1_1} and \Cref{itm:thm:comp_llt_C=1_2}, for each $j\in[k]$ we define $T_{r_j}\subseteq S_{r_j}$ to be of size $n''$ and contain indices satisfying \Cref{itm:thm:comp_llt_C=1_2}.

Recall that $\phi$ is the least common multiple of values in $[t]\setminus\Phi$.
Now we show that the sum of $t\phi\cdot k$ random variables ($t\phi$ from each one of $T_{r_1},\ldots,T_{r_k}$) is a random variable that satisfies the conditions in \Cref{lem:comp_llt_simp_C=1}.
Formally, let $m=\floorbra{n''/(t\phi)}$ and arbitrarily select $m$ disjoint subsets $T_{r_j}^1,\ldots,T_{r_j}^m$ of size $t\phi$ from each $T_{r_j}$.
Define random variables
$$
Y_\ell=\sum_{j\in[k]}\sum_{i\in T_{r_j}^\ell}X_i
\quad\text{for each $\ell\in[m]$}
$$
and define 
$$
Y_0=\sum_{i\notin\bigcup_{j\in[k],\ell\in[m]}T_{r_j}^\ell}X_i
$$ 
to be the sum of the remaining $X_i$'s.
We will show that for each $\ell\in[m]$, there exists $u_\ell\in\Zbb$ such that both $\Pr[Y_\ell=u_\ell]$ and $\Pr[Y_\ell=u_\ell+\phi]$ are at least $\alpha=\pbra{\frac L{4n(t+1)}}^{t^2\phi}$.
Then \Cref{thm:comp_llt_C=1} follows from \Cref{lem:comp_llt_simp_C=1} by conditioning on $Y_0$ and observing $m=\floorbra{\ceilbra{\floorbra{L/(4t)}/t^2}/(t\phi)}\ge\floorbra{L/(16t^4\phi)}$.

Fix an arbitrary $\ell\in[m]$ and define $w_j=z_{r_j}-z_{r_j}'$ for each $j\in[k]$.
By \Cref{itm:thm:comp_llt_C=1_1}, $|w_j|\le t$ and $r_j$ does not divide it.
Hence the greatest common divisor $g$ of $|w_1|,\ldots,|w_k|$ lies in $[t]\setminus\Phi$, which must divide $\phi$.
Thus by B\'ezout's identity (see e.g., \cite{wiki:Bézout's_identity}), there exist $s_1,\ldots,s_k\in\Zbb$ such that 
\begin{equation}\label{eq:thm:comp_llt_C=1_2}
\sum_{j\in[k]}s_j\cdot w_j=\phi.
\end{equation}
In addition, we can assume that $|s_j|\le\phi/g\cdot\max_{j\in[k]}|w_j|/g\le t\phi$ \cite{bezout_mo}.
Now we define $u_\ell$ as
$$
u_\ell=\sum_{j\in[k]: s_j<0}z_{r_j}\cdot t\phi+\sum_{j\in[k]: s_j\ge0}z_{r_j}'\cdot t\phi.
$$
Then the probability of $Y_\ell=u_\ell$ is at least the probability that every $X_i\in T_{r_j}^\ell$ equals $z_{r_j}$ if $s_j<0$, and every $X_i\in T_{r_j}^\ell$ equals $z_{r_j}'$ if $s_j\ge0$.
Hence by \Cref{itm:thm:comp_llt_C=1_2} and the independence of the $X_i$'s, we have $\Pr\sbra{Y_\ell=u_\ell}\ge\pbra{\frac L{4n(t+1)}}^{t\phi\cdot k}\ge\alpha$ as desired.
To analyze $u_\ell+\phi$, we rewrite it as 
\begin{align*}
u_\ell+\phi
&=\sum_{j\in[k]: s_j<0}\pbra{z_{r_j}\cdot t\phi+s_j\cdot w_j}+\sum_{j\in[k]: s_j\ge0}\pbra{z_{r_j}'\cdot t\phi+s_j\cdot w_j}
\tag{by \Cref{eq:thm:comp_llt_C=1_2}}\\
&=\sum_{j\in[k]: s_j<0}\pbra{z_{r_j}'\cdot|s_j|+z_{r_j}\cdot(t\phi-|s_j|)}+\sum_{j\in[k]: s_j\ge0}\pbra{z_{r_j}\cdot|s_j|+z_{r_j}'\cdot(t\phi-|s_j|)}.
\tag{since $w_j=z_{r_j}-z_{r_j}'$}
\end{align*}
Hence the probability of $Y_\ell=u_\ell+\phi$ is at least the probability that $|s_j|$ (resp., $t\phi-|s_j|$) many $X_i\in T_{r_j}^\ell$ equal $z_{r_j}'$ (resp., $z_{r_j}$) if $s_j<0$, and $|s_j|$ (resp., $t\phi-|s_j|$) many $X_i\in T_{r_j}^\ell$ equal $z_{r_j}$ (resp., $z_{r_j}'$) if $s_j\ge0$.
Therefore $\Pr\sbra{Y_\ell=u_\ell+\phi}\ge\alpha$ follows again from \Cref{itm:thm:comp_llt_C=1_2} and the independence of $X_i$'s.
\end{proof}

\end{document}